\documentclass[draftcls,one column,11pt]{IEEEtran}

\usepackage{amsmath,amssymb,amsfonts}
\usepackage{color}
\usepackage{graphicx}
\usepackage{dsfont}

\newtheorem{lemma}{\textbf{Lemma}}
\newtheorem{corollary}{\textbf{Corollary}}
\newtheorem{definition}{\textbf{Definition}}
\newtheorem{theorem}{\textbf{Theorem}}
\newtheorem{proposition}{\textbf{Proposition}}

\newtheorem{observation}{\textbf{Observation}}

\newcommand{\ie}{\emph{i.e.}}
\newcommand{\eg}{\emph{e.g.}}
\newcommand{\etal}{\emph{et al. }}

\newcommand{\gaussnum}[2]{{#1 \brack #2}}
\newcommand{\sspan}[1]{\left\langle #1 \right\rangle}
\newcommand{\Prob}[1]{{\ensuremath{\mathbb{P}\left[ #1 \right]} }}
\newcommand{\Expc}[2]{{\ensuremath{\mathbb{E}_{#1}\left[ #2 \right]} }}

\renewcommand{\=}{\doteq}
\newcommand{\1}{\mathds{1}}
\newcommand{\col}{\mbox{col}}

\newcommand{\Fbb}{\mathbb{F}}

\newcommand{\Cca}{\mathcal{C}}

\newcommand{\Rca}{\mathcal{R}}
\newcommand{\Sca}{\mathcal{S}}
\newcommand{\Xca}{\mathcal{X}}
\newcommand{\Yca}{\mathcal{Y}}

\newcommand{\aleq}{\stackrel{.}{\leq}}
\newcommand{\wt}{\widetilde}

\title{On the Capacity of Non-Coherent\\ Network Coding}

\author{M.~Jafari~Siavoshani, S.~Mohajer, C.~Fragouli, S
  N.~Diggavi\\ \vspace{0.2in} Ecole Polytechnique F\'{e}d\'{e}rale de Lausanne
  (EPFL), Switzerland\\
University of California, Los Angeles (UCLA), USA  
\thanks{S N. Diggavi was at EPFL and now is at the
    University of California, Los Angeles (UCLA). The work of
    M.~Jafari~Siavoshani and C.~Fragouli was supported in part by the
    Swiss National Science Foundation through the grant
    \#~PP002-110483. The work of S.~Mohajer and C.~Fragouli was
    supported in part by the ERC Starting Investigator grant
    \#~240317. 
    }}

\begin{document}
\maketitle

\begin{abstract}

  We consider the problem of multicasting information from a source to
  a set of receivers over a network where intermediate network nodes
  perform randomized network coding operations on the source packets.
  We propose a channel model for the non-coherent network coding
  introduced by Koetter and Kschischang in
  \cite{KoetKsch-IT08-erasure}, that captures the essence of such a
  network operation, and calculate the capacity as a function of
  network parameters. We prove that use of subspace coding is optimal,
  and show that, in some cases, the capacity-achieving distribution
  uses subspaces of several dimensions, where the employed dimensions
  depend on the packet length.  This model and the results also allow
  us to give guidelines on when subspace coding is beneficial for the
  proposed model and by how much, in comparison to a coding vector
  approach, from a capacity viewpoint. We extend our results to the
  case of multiple source multicast that creates a virtual multiple
  access channel\footnote{Some parts of the work in this paper was
    presented at ISIT'08, ISIT'09, and ITW'09.}.
\end{abstract}

\centerline{\textbf{Keywords}}
Network coding, non-coherent communication, subspace coding, channel capacity,
multi-source multicast, randomized network coding.

\section{Introduction}

The network coding techniques for information transmission in networks
introduced in \cite{AhlCaiLiYeu} have attracted significant interest
in the literature, both because of posing theoretically interesting
questions, as well as because of potential impact in applications. The
first fundamental result proved in network coding, and perhaps still
the most useful from a practical point of view today, is that, using
linear network coding \cite{LiCaYe-IT03,KoMe-NetTr03}, one can achieve
rates up to the common min-cut value when multicasting to $N_r\geq 1$
receivers. In general this may require operations over a field of size
approximately $\sqrt{N_r}$, which translates to communication using
packets of length $\frac{1}{2}\log N_r$ bits \cite{information_flow}.

However, this result assumes that the receivers know perfectly the
operations that the network nodes perform. In large dynamically
changing networks, collecting network information comes at a cost, as
it consumes bandwidth that could instead have been used for
information transfer. In practical networks, where such deterministic
knowledge is not sustainable, the most popular approach is to perform
randomized network coding \cite{HoKoMeEfShKa2006} and to append coding
vectors at the headers of the packets to keep track of the linear
combinations of the source packets they contain (see, \emph{e.g.},
\cite{ChWuJa-Alert03}). The coding vectors have an overhead of $h \log
N_r$ bits, where $h$ is the total number of packets to be linearly
combined.  This results in a loss of information rate that can be
significant with respect to the min-cut value. In particular, in
wireless networks such as sensor networks where communication is
restricted to short packet lengths, the coding vector overhead can be
a significant fraction of the overall packet length \cite{TinyOs,KeJaFrArDi-InfoCom09}.

Use of coding vectors is akin to use of training symbols to learn the
transformation induced by a network. A different approach is to assume
a non-coherent scenario for communication, as proposed in \cite{KoetKsch-IT08-erasure}, where neither the source(s)
nor the receiver(s) have any knowledge of the network topology or the
network nodes operations. Non-coherent communication allows for
creating end-to-end systems completely oblivious to the network state.
Several natural questions arise considering this non-coherent framework: 
{\sf (i)} what are the fundamental limits on the rates that can be 
achieved in a network where the intermediate node operations
are unknown, {\sf (ii)} how can they be achieved, and {\sf (iii)} how do 
they compare to the coherent case.

In this work we address such questions for two different cases. First,
we consider the scenario where a single source aims to transmit
information to one or multiple receiver(s) over a network under the
non-coherence assumption using fixed packet length. Because network
nodes only perform linear operations, the overall network behavior
from the source(s) to a receiver can be represented as a matrix
multiplication of the sent source packets.  We consider operation in
time-slots, and assume that the channel transfer matrices are
distributed uniformly at random and i.i.d. over different time-slots.
Under this probabilistic model, we characterize the asymptotic 
capacity behavior of the introduced channel and show that using
\emph{subspace coding} we can achieve the optimal performance. We extend our
model for the case of multiple sources and characterize the asymptotic
behavior of the optimal rate region for the case of two sources. We
believe that this result can be extended to the case of more
than two sources using the same method that is applied in \S\ref{sec:ProofMultiSource}. 
For the multi-source case we prove as well that
encoding information using subspaces is sufficient to achieve the
optimal rate region.

The idea of non-coherent modeling for randomized network coding was first
proposed in the seminal work by Koetter and Kschischang in
\cite{KoetKsch-IT08-erasure}. In that work, the authors focused on
algebraic subspace code constructions over a
Grassmannian. Independently and in parallel to our work
in \cite{JaFrDi-ISIT08}, Montanari \etal \cite{MoUr-07} introduced a
different probabilistic model to capture the end-to-end functionality
of non-coherent network coding operation, with a focus on the case of
error correction capabilities. Their model does not examine subsequent
time slots, but instead, allows the packets block length (in this
paper terminology; packet length $T$) to increases to infinity, with
the result that the overhead of coding vectors becomes negligible,
very fast.

Silva \etal \cite{SiKschKo-CapRandNetCod} independently and subsequent
to our works in \cite{JaFrDi-ISIT08} and \cite{JaMoFrDi-ISIT09}, also considered a probabilistic
model for non-coherent network coding, which is an extension of the
model introduced in \cite{MoUr-07} over multiple time-slots. In their
model the transfer matrix is constrained to be square as well as full
rank.  This is in contrast to our model, where the transfer matrix can
have arbitrary dimensions, and the elements of the transfer matrix are
chosen uniformly at random, with the result that the transfer matrix
itself may not have full rank (this becomes more pronounced for small
matrices).  
Moreover, we extend our work to multiple source multicast, which
corresponds to a virtual non-coherent multiple access channel (MAC).  Our
results coincide for the case of a single source, when the packet
length and the finite field of operations are allowed to grow
sufficiently large. Another difference is that the work in
\cite{SiKschKo-CapRandNetCod} focuses on additive error with constant
dimensions; in contrast, we focus on  packet erasures.

An interpretation of our results is that it is the finite field analog
of the Grassmannian packing result for non-coherent MIMO channels as
studied in the well known work in \cite{ZhTse-IT02}. In particular, we
show that for the non-coherent model over finite fields, the capacity
critically depends on the relationship between the ``coherence time''
(or packet length $T$ in our model) and the min-cut of the network. In fact the number
of active subspace dimensions depend on this relationship; departing
from the non-coherent MIMO analogy of \cite{ZhTse-IT02}.

The paper is organized as follows. We define our notation and channel
model in \S\ref{sec:ChannelModel-Notations}; we state and discuss our
main results in \S\ref{sec:MainResults}; we prove the capacity results
for the single and multiple sources in sections
\S\ref{sec:TheChanlCap-SnglSrc} and \S\ref{sec:ProofMultiSource}
respectively; and conclude the paper in \S\ref{sec:Conclusion}. 

All the missing proofs for lemmas, theorems, and etc., are given in
Appendix~\ref{sec:apndx1} unless otherwise stated.

\section{Channel Model and Notation}\label{sec:ChannelModel-Notations}

\subsection{Notation}\label{subsec:Notations}

We here introduce the notation and definitions we use in the 
following sections. Let $q\ge 2$ be a power of a prime. In this 
paper, all vectors and matrices have elements in a finite field 
$\Fbb_q$. We use $\Fbb^{n\times m}_q$ to denote the set of all 
$n\times m$ matrices over $\Fbb_q$, and $\Fbb^T_q$ to denote 
the set of all row vectors of length $T$.  The set $\Fbb_q^T$ 
forms a $T$-dimensional vector space over the field $\Fbb_q$.

Throughout the paper, we  use capital letters, \eg, $X$, to 
denote random objects, including random variables, random 
matrices, or random subspaces, and corresponding lower-case 
letters, \eg, $x$ to denote their realizations. For example, we 
denote by $\Pi$ a ``random subspace'' which takes as values  
the subspaces in a vector space according to some distribution, 
and by $\pi$ a specific realization. 
Also, bold capital letters, \eg, $\mathbf{A}$, are reserved 
for deterministic matrices and bold lower-case 
letters, \eg, $\mathbf{v}$, are used for deterministic vectors.

For subspaces $\pi_1$ and $\pi_2$, $\pi_1 \sqsubseteq \pi_2$  
denotes that $\pi_1$ is a subspace of $\pi_2$. Recall that for 
two subspaces $\pi_1$ and $\pi_2$, $\pi_1\cap \pi_2$ is the 
intersection of these subspaces which itself is a subspace. 
We use $\pi_1 + \pi_2$ to denote the smallest subspace that 
contains both $\pi_1$ and $\pi_2$, namely,
\begin{align*}
\pi_1 + \pi_2=\left\{\mathbf{v}_1+\mathbf{v}_2| \mathbf{v}_1\in \pi_1, \mathbf{v}_2\in \pi_2 \right\}. 
\end{align*}
It is well known that
\begin{align*}
\dim(\pi_1+\pi_2)=\dim(\pi_1)+\dim(\pi_2)-\dim(\pi_1 \cap \pi_2).
\end{align*}

For a set of vectors $\{\mathbf{v}_1,\ldots,\mathbf{v}_k\}$ we denote their 
linear span by $\sspan{\mathbf{v}_1,\dots,\mathbf{v}_k}$. For a matrix 
$\mathbf{X}$, $\sspan{\mathbf{X}}$ is the subspace spanned by the rows 
of $\mathbf{X}$ and $\sspan{\mathbf{X}}_c$ is the subspace spanned by 
the columns of $\mathbf{X}$. We then  have 
$\mathrm{rank}(\mathbf{X})=\dim(\sspan{\mathbf{X}})=\dim(\sspan{\mathbf{X}}_c)$.

We use the calligraphic symbols, \ie, $\Xca$ or $\Yca$ to denote a set 
of matrices. To denote a set of subspaces we use the same calligraphic 
symbols but with a ``$\sim$'', \ie, $\wt{\Xca}$ or $\wt{\Yca}$.

We use the symbols ``$\succ$'' and ``$\prec$'' to denote the element-wise 
inequality between vectors and matrices of the same size.

For two real valued functions $f(x)$ and $g(x)$ of $x$, we use $f(x)\=g(x)$ 
to denote that\footnote{One has to specify the growing variable whenever ``$\doteq$'' 
is used for multi-variate functions.  However, since in this work the growing 
variable is always $q$, the field size, we will not repeat it for sake of brevity.}

\[
\lim_{x\rightarrow\infty}\frac{\log f(x)}{\log g(x)} \rightarrow 1.
\]
Note that the definition of ``$\=$'' is different from the more 
standard definition which is 
\mbox{$\lim_{x\to\infty}\frac{1}{x}\log\frac{f(x)}{g(x)}\to 0$.}
We also use a similar definition for $f\aleq g$ to denote that
\[
\lim_{x\rightarrow\infty}\frac{\log f(x)}{\log g(x)} \rightarrow c\le 1,
\]
where $c$ is a constant.

We use the big-$O$ notation which is defined as follows. Let $f(x)$ and $g(x)$ be two 
functions defined on some subset of the real numbers. We write
$f(x)=O\left(g(x)\right)\mbox{ as } x\to\infty,$
if there exists a positive real number $M$ and a real number $x_0$ such that
$|f(x)| \le M |g(x)| \mbox{ for all } x>x_0.$ 
For the little $o$ notation we use the following definition. We write
$f(x)=o(g(x))\mbox{ as } x\to\infty,$
if for all $\epsilon>0$ there exists a real number $x_0$ such that
$|f(x)| \le \epsilon \cdot |g(x)| \mbox{ for all } x>x_0.$
We use also the big-$\Omega$ notation which is defined as follows. We write
$f(x)=\Omega\left(g(x)\right)\mbox{ as } x\to\infty,$ if we have
$g(x)=O\left(f(x)\right)\mbox{ as } x\to\infty$.
Finally, we use the big-$\Theta$ notation to denote that a function is bounded both above 
and below by another function asymptotically. Formally, we write
$f(x)=\Theta\left(g(x)\right)\mbox{ as } x\to\infty,$ if and only if we have
$f(x)=O\left(g(x)\right)$ and $f(x)=\Omega\left(g(x)\right)\mbox{ as } x\to\infty$.

\begin{definition} [{Grassmannian and Gaussian coefficient \cite{VanLintWilson-Combinatorics,And-Book76}}]
The Grassmannian $\mathrm{Gr}(T,d)_q$ is the set of all $d$-dimensional 
subspaces of the $T$-dimensional space over a finite field  $\Fbb_q$, 
namely,
\begin{align*}
\mathrm{Gr}(T,d)_q\triangleq\{\pi\sqsubseteq \Fbb_q^T: \dim(\pi)=d\}.
\end{align*}
The cardinality of $\mathrm{Gr}(T,d)_q$ is the Gaussian coefficient, namely, 
\begin{align}\label{eq:GaussianNumber}
\gaussnum{T}{d}_q &\triangleq |\mathrm{Gr}(T,d)_q| = \frac{(q^T-1)\cdots(q^{T-d+1}-1)}{(q^d-1)\cdots(q-1)}.
\end{align}
\end{definition}

\begin{definition}[{The set $\mathrm{Sp}(T,m)_q$}]
We define $\mathrm{Sp}(T,m)_q$ to be the set (sphere) of all subspaces 
of dimension at most $m$ in the $T$-dimensional space $\Fbb_q^T$, namely
\begin{align}\nonumber
\mathrm{Sp}(T,m)_q\triangleq\bigcup_{d=0}^{\min[m,T]} \mathrm{Gr}(T,d)_q = \{\pi\sqsubseteq \Fbb_q^T: \dim(\pi)\leq \min[m,T]\}.
\end{align}
The cardinality of $\mathrm{Sp}(T,m)_q$ equals
\[ \Sca(T,m)_q \triangleq |\mathrm{Sp}(T,m)_q| =\sum_{d=0}^{\min[m,T]} |\mathrm{Gr}(T,d)_q|.\]
\end{definition}

\begin{definition}[{The number  $\psi(T,n,\pi_d)_q$} ]\label{def:psi}
We denote by $\psi(T,n,\pi_d)_q$ the number of different $n\times T$ 
matrices with elements from a field $\Fbb_q$, such that their rows 
span a specific subspace $\pi_d\sqsubseteq\Fbb_q^T$ of dimension 
$0\le d\le\min[n,T]$.
\end{definition}

For simplicity, in the rest of the paper we will drop the subscript 
$q$ in the previous definitions whenever it is obvious from the context.

\subsection{Preliminary Lemmas}\label{sec:PrelimLemma}

We here state some preliminary lemmas related to the definitions 
introduced in \S\ref{subsec:Notations}.

Existing bounds in the literature allow to approximate the Gaussian number,
for example, we have from \cite[Lemma~4]{KoetKsch-IT08-erasure} that
\cite[Section~III]{GadYan-IT08}
\begin{equation}\label{eq:GaussianNumber_bound_1}
q^{d(T-d)} < \gaussnum{T}{d} < \frac{q^{d(T-d)}}{\prod_{j=1}^{\infty} (1-q^{-j})} < 4 q^{d(T-d)},\quad \forall d: 0<d<T.
\end{equation}
Using Definition~\ref{eq:GaussianNumber} and \eqref{eq:GaussianNumber_bound_1}
we have Lemma~\ref{lem:gauss_Approximation}.
\begin{lemma}\label{lem:gauss_Approximation}
For large $q$ we can approximate the Gaussian number as follows
\[
\gaussnum{T}{d} = q^{d(T-d)}(1+O(q^{-1})) \= q^{d(T-d)}.
\]
\end{lemma}

\begin{lemma}\label{lem:psi_value}
For $\psi(T,n,\pi_d)$ given in Definition~\ref{def:psi}, we have that \cite{Gab-PIT85}
\begin{align*}
\psi(T,n,\pi_d) = \prod_{i=0}^{d-1} (q^n-q^i) = q^{\binom{d}{2}} \prod_{i=0}^{d-1} (q^{n-i}-1),
\end{align*}
\ie, it does not depend on $T$.
\end{lemma}
Since $\psi(T,n,\pi_d)$ does not depend on $T$, and only depends on $\pi_d$ through 
its dimension, as a shorthand notation we will also use $\psi(n,d)$ instead of 
$\psi(T,n,\pi_d)$, where $d=\dim(\pi_d)$.

Using Lemma~\ref{lem:psi_value} the following lower and upper bounds
are straightforward
\begin{equation}\label{eq:psi_bound}
(1-dq^{-n+d-1}) < \left(1-\sum_{i=0}^{d-1}q^{-n+i}\right) < \frac{\psi(n,d)}{q^{nd}} < 1,
\end{equation}
which imply Lemma~\ref{lem:psi_Approximation} (see also \cite{GadYan-IT08}). 
\begin{lemma}\label{lem:psi_Approximation}
For large values of $q$ the following approximation holds
\[
\psi(n,d) = q^{nd}(1+O(q^{-1})) \= q^{nd}.
\]
\end{lemma}

It is also worthwhile to mention that $\psi(n,d) \gaussnum{T}{d}$ is the number of
$n \times T$ matrices of rank $d$. We can count all the $n\times T$ matrices through the following
 Lemma~\ref{lem:psi_Recursive},
(also see \cite{VanLintWilson-Combinatorics,And-Book76}, and \cite[Corollary~5]{Gab-PIT85}).
\begin{lemma}\label{lem:psi_Recursive}
For every $n>0$ and $T>0$ we can write
\begin{equation*}
\sum_{d=0}^{\min[n,T]} \psi(n,d) \gaussnum{T}{d} = q^{nT},
\end{equation*}
where $\psi(n,0)=1$.
\end{lemma}

\subsection{The Non-Coherent Finite Field Channel Model}

We consider a network where nodes perform random linear 
network coding over a finite field $\Fbb_q$. 
We are interested in the maximum information rate at which a 
single (or multiple) source(s) can successfully  communicate
over such a network when neither the transmitter nor the 
receiver(s) have any channel state information (CSI). 
For simplicity, we will present the channel model and our
analysis for the case of a single receiver; the extension 
to multiple receivers (with the same channel parameters) 
is straightforward, as we  also discuss 
in the results section.

We assume that time is slotted and the channel is block 
time-varying. For the single source communication, at time 
slot $t$, the receiver observes
\begin{equation}\label{eq:channel_model_P2P}
Y[t]=G[t]X[t], 
\end{equation}
where $X[t]\in\Fbb_q^{m\times T}$, $G[t]\in\Fbb_q^{n\times m}$, 
and $Y[t]\in\Fbb_q^{n\times T}$. At each time-slot, the receiver 
receives $n$ packets of length $T$ (captured by the rows of 
matrix $Y[t]$) that are random linear combinations of the $m$ packets 
injected by the source (captured by the rows of matrix $X[t]$). 
In our model, the packet length $T$ can be interpreted as the 
coherence time of the channel, during which the transfer matrix 
remains constant. Each element of the transfer matrix $G[t]$ is 
chosen uniformly at random from $\mathbb{F}_q$, changes independently 
from time slot to time slot, and is unknown to both the source and 
the receiver. In other words, the channel transfer matrix is 
chosen uniformly at random from all possible matrices in 
$\Fbb_q^{n\times m}$ and has i.i.d. distribution over different 
blocks. In general, the topology of the network may impose some 
constraints on the transfer matrix $G[t]$ (for example, some 
entries might be zero, see \cite{KoMe-NetTr03,JaFrDi-ITW07,SaMaFr-NetCode2009,ShJaDe-ITA09}). 
However, we believe that this is a reasonable general 
model, especially for large-scale dynamically-changing networks 
where apart from random coefficients there exist many other 
sources of randomness. Formally, we define the non-coherent 
matrix channel as follows.

\begin{definition}[Non-coherent matrix channel $\mathrm{Ch}_{\textsl{m}}$]\label{def:matrix_channel_P2P}
This is defined to be the matrix channel $\mathrm{Ch}_{\textsl{m}}:\Xca\rightarrow \Yca$ described by (\ref{eq:channel_model_P2P}) with the assumption that $G[t]$ is i.i.d. and uniformly distributed over all matrices $\Fbb_q^{n\times m}$.  It is a discrete memoryless channel with input alphabet $\mathcal{X}\triangleq\Fbb_q^{m\times T}$ and output alphabet $\mathcal{Y}\triangleq\mathbb{F}_q^{n\times T}$.
\end{definition}
The capacity of the channel $\mathrm{Ch}_{\textsl{m}}$ is given by 
\begin{equation}\label{eq:capacity_def}
C_{\textsl{m}} = \max_{P_X(x)} I(X;Y),
\end{equation}
where $P_X(x)$ is the input distribution. To achieve the capacity 
a coding scheme may employ the channel given in (\ref{eq:channel_model_P2P}) 
multiple times, and a codeword is a sequence of input matrices 
from $\mathcal{X}$. For a coding strategy that induces an 
input distribution $P_X(x)$, the achievable rate is
\[
R = I(X;Y).
\]

Now we define a non-coherent subspace channel $\mathrm{Ch}_{\textsl{s}}$ 
which takes as an input a subspace 
and outputs another subspace. Then, in Theorem~\ref{thm:channel_equivalence_P2P} we will show that the
two channels $\mathrm{Ch}_{\textsl{m}}$ and $\mathrm{Ch}_{\textsl{s}}$ are
equivalent from the point of view of calculating the mutual information
between their inputs and their outputs.
\begin{definition}[Non-coherent subspace channel $\mathrm{Ch}_{\textsl{s}}$]\label{def:subspace_channel_P2P}
This is defined to be the channel $\mathrm{Ch}_{\textsl{s}}:\wt{\Xca}\rightarrow \wt{\Yca}$ with input alphabet $\wt{\Xca}=\mathrm{Sp}(T,m)$ and output alphabet $\wt{\Yca}=\mathrm{Sp}(T,n)$ and transition probability
\begin{equation}\label{eq:P2P_channel_transfer_prob_2}
P_{\Pi_Y|\Pi_X}(\pi_y|\pi_x) \triangleq \left\{ \begin{array}{ll} \psi(T,n,\pi_y) q^{-n\dim(\pi_x)} & \pi_y\sqsubseteq \pi_x,\\ 
0 & \text{otherwise},
\end{array} \right.
\end{equation}
where $\Pi_X$ and $\Pi_Y$ are the input and output variables 
of the channel $\mathrm{Ch}_{\textsl{s}}$.
\end{definition}
The capacity of the channel $\mathrm{Ch}_{\textsl{s}}$ is given by
\[
C_{\textsl{s}} = \max_{P_{\Pi_X}(\pi_x)} I(\Pi_{X};\Pi_Y),
\]
where $P_{\Pi_X}(\pi_x)$ is the input distribution defined over the
set of subspaces $\wt\Xca$.

We  next consider a multiple sources scenario, and the multiple 
access channel corresponding to (\ref{eq:channel_model_P2P}). 
In this case, we have 
\begin{equation}\label{eq:channel_model_MAC} 
Y[t]=\sum_{i=1}^{N_s} G_i[t]X_i[t],
\end{equation}
where $N_s$ is the number of sources, and each source $i$ inserts $m_i$ packets 
to the network. Thus, $X_i[t]\in\Fbb_q^{m_i\times T}$, $G_i[t]\in\Fbb_q^{n\times m_i}$ 
and $Y[t]\in\Fbb_q^{n\times T}$. We can also collect all $G_i[t]$ 
in an $n\times \sum_{i=1}^{N_s} m_i$ matrix $G_{\textsl{MAC}}[t]$ 
and all $X_i[t]$ in an $\sum_{i=1}^{N_s} m_i\times T$ matrix 
$X_{\textsl{MAC}}[t]$ as following
\[
X_{\textsl{MAC}}[t] = \left[\begin{array}{c}
X_1[t]\\
\vdots\\
X_{N_s}[t]\\
\end{array} \right],\quad\text{and}\quad 
G_{\textsl{MAC}}[t] = \left[\begin{array}{ccc} G_1[t] & \cdots & G_{N_s}[t] \end{array}\right],
\]
so we can rewrite \eqref{eq:channel_model_MAC} as 
\begin{equation*} 
Y[t]=G_{\textsl{MAC}}[t] X_{\textsl{MAC}}[t].
\end{equation*}
Each source $i$ then controls $m_i$ rows of the matrix 
$X_{\textsl{MAC}}[t]$. Again we assume that each  entry 
of the matrices $G_i[t]$ is chosen i.i.d. and uniformly at 
random from the field $\Fbb_q$ for all source nodes and all time 
instances.

\begin{definition}[The non-coherent multiple access matrix 
channel $\mathrm{Ch}_{\textsl{m-MAC}}$]\label{def:matrix_channel_MAC}
This is defined to be the channel 
\mbox{$\mathrm{Ch}_{\textsl{m-MAC}}:\Xca_1\times\cdots\times \Xca_{N_s}\rightarrow\Yca$} 
described in \eqref{eq:channel_model_MAC}, with the assumption 
that $G_i[t]$, $i=1,\ldots,N_s$, are i.i.d. and uniformly 
distributed over all matrices $\Fbb_q^{n\times m_i}$, $i=1,\ldots,N_s$. 
It forms a discrete memoryless MAC  with input 
alphabets $\Xca_i\triangleq\Fbb_q^{m_i\times T}$, $i=1,\ldots,N_s$, and 
output alphabet $\mathcal{Y}\triangleq\mathbb{F}_q^{n\times T}$.
\end{definition}
It is well known \cite{CoTh-ElmntsInfoTheory06} that the rate 
region of any multiple access channel including 
$\mathrm{Ch}_{\textsl{m-MAC}}$ is given by the closure of 
the convex hull of the rate vectors satisfying
\begin{equation*}
R_S\le I(X_S;Y|X_{S^c}) \quad \text{for all } S\subseteq\{1,\ldots,N_s\},
\end{equation*}
for some product distribution $P_{X_1}(x_1)\cdots P_{X_{N_s}}(x_{N_s})$. 
Note that $R_S=\sum_{i\in S} R_i$ where $R_i$ is the transmission rate 
of the $i$th source, $X_S=\{X_i:i\in S\}$ and $S^c$ is the complement 
set of $S$.

As before, we define a non-coherent subspace version\footnote{For simplicity, 
we restrict this definition to only two source nodes. However, generalization 
 to  $N_s$ sources is straightforward.} of the matrix multiple 
access channel and in Theorem~\ref{thm:channel_equivalence_MAC} 
we show that from the point of view of rate region these two channels 
are equivalent.
\begin{definition}[Non-coherent subspace multiple access channel $\mathrm{Ch}_{\textsl{s-MAC}}$]\label{def:subspace_channel_MAC}
This is defined to be the channel $\mathrm{Ch}_{\textsl{s-MAC}}:\wt{\Xca}_1\times\wt{\Xca}_2\rightarrow \wt{\Yca}$ with input alphabets $\wt{\Xca}_i=\mathrm{Sp}(T,m_i)$, $i=1,2$, output alphabet $\wt{\Yca}=\mathrm{Sp}(T,n)$ and transition probability
\begin{align}\label{eq:MAC_channel_transfer_prob_2}
\Pr(\Pi_Y=\pi_y |& \Pi_{X_1}=\pi_1, \Pi_{X_2}=\pi_2)
=\left\{ \begin{array}{ll}
\psi(T, n, \pi_y) q^{-n\dim( \pi_1 + \pi_2)} & \pi_y\sqsubseteq \pi_1 + \pi_2,\\
0 & \text{otherwise},
\end{array} \right.
\end{align}
where $\Pi_{X_1}$ and $\Pi_{X_2}$ are the input and $\Pi_Y$ is the 
output variables of the channel $\mathrm{Ch}_{\textsl{s-MAC}}$.
\end{definition}

\section{Main Results}\label{sec:MainResults}

\subsection{Single Source}
Our main results, Theorem~\ref{thm:Main_Result_Single_Src} and
Theorem~\ref{thm:Main_Result_Single_Src_Dist}, characterize the
capacity for non-coherent network coding for the model given in
\eqref{eq:channel_model_P2P}. We show that the capacity is achieved
through subspace coding, where the information is communicated from
the source to the receivers through the choice of subspaces. Formally,
we have the following results.

\begin{theorem}\label{thm:channel_equivalence_P2P}
The matrix channel $\mathrm{Ch}_{\textsl{m}}:\Xca\rightarrow \Yca$ 
defined in Definition~\ref{def:matrix_channel_P2P} and the subspace 
channel $\mathrm{Ch}_{\textsl{s}}:\wt{\Xca}\rightarrow\wt{\Yca}$ defined
in Definition~\ref{def:subspace_channel_P2P}  are equivalent 
in terms of evaluating the mutual information 
between the input and output. 
 More precisely, for every input distribution for the channel $\mathrm{Ch}_{\textsl{s}}$ 
there is an input distribution for the channel $\mathrm{Ch}_{\textsl{m}}$ 
such that $I(X;Y)=I(\Pi_X;\Pi_Y)$ and vice versa.
As a result, these channels have the same capacity $C_{\textsl{m}}=C_{\textsl{s}}$. 
\end{theorem}

For the proof of Theorem~\ref{thm:channel_equivalence_P2P} refer
to Appendix~\ref{sec:apndx1} and for more discussion refer to
\S\ref{subsec:Equiv-MatrixChannel-SubspaceChannel}.

\begin{theorem}\label{thm:Main_Result_Single_Src}
For the channel $\mathrm{Ch}_{\textsl{m}}:\Xca\rightarrow \Yca$ 
defined in Definition~\ref{def:matrix_channel_P2P},
the capacity is given by
\begin{align}
C_{\textsl{m}} &= i^*(T-i^*)\log_2 q + o(1),
\end{align}
where $i^* = \min\left[m,n,\lfloor T/2\rfloor\right]$, and $o(1)$ tends to zero as $q$ grows.
\end{theorem}
Theorem~\ref{thm:Main_Result_Single_Src} is proved in
\S\ref{subsec:Capacity-UpperLowerBound}.  The result of
Theorem~\ref{thm:Main_Result_Single_Src} is for large alphabet
regime\footnote{We gratefully acknowledge the contribution of an
  anonymous reviewer who gave an alternate proof, which
  focused on the asymptotic $q$ regime. We have included that proof
  in \S\ref{subsec:Capacity-UpperLowerBound}. Our original proof was
  based partially on the proof now given for
  Theorem~\ref{thm:Main_Result_Single_Src_Dist}.}. The following
result, Theorem~\ref{thm:Main_Result_Single_Src_Dist}, 
 is valid for a finite field size, and therefore is a
non-asymptotic result.

\begin{theorem}\label{thm:Main_Result_Single_Src_Dist}
Consider the channel $\mathrm{Ch}_{\textsl{m}}:\Xca\rightarrow \Yca$ 
defined in Definition~\ref{def:matrix_channel_P2P}. 
There exists a finite number~$q_0$ such that for $q>q_0$
the optimal input distribution is nonzero only for matrices of rank 
 in the  set
\begin{equation}
\mathcal{A} = \left\{\min\left[(T-n)^+,m,n,T\right],\ldots,\min\left[m,n,T\right] \right\}.
\end{equation}
%
Moreover, for all values of $q$ the optimal input distribution is uniform over all matrices $X$ of 
the same rank, and the total probability allocated to transmitting matrices of rank $i$ 
equals
\begin{equation}
\alpha^*_i \triangleq  \Prob{\mathrm{rank}(X)=i} = 2^{-C_{\textsl{m}}} q^{i(T-i)}
\left[1+o(1)\right],\quad \forall
i\in\mathcal{A}.
\end{equation}

\end{theorem}
The proof of Theorem~\ref{thm:Main_Result_Single_Src_Dist} is presented
in \S\ref{subsec:OptSolution_SingleSrc_GenApproach} and \S\ref{subsec:OptSolution_SingleSrc_Large_q}, 
and uses standard techniques from convex optimization, as well as 
large field size approximations. Note that, the same 
coding scheme at the source simultaneously achieves the capacity
for all receivers with the same channel parameters (\ie, values of $n$, $m$ and $T$).  
That is, each receiver is able to successfully decode.

The result of Theorem~\ref{thm:Main_Result_Single_Src_Dist} for the
active set of input dimensions is not asymptotic in $q$. However, it
is not easy to analytically find the minimum value of $q_0$  such that
the theorem statement holds for all $q>q_0$.
Theorem~\ref{thm:Main_Result_Single_Src_Extnd} demonstrates how we can
analytically characterize $q_0$ given in
Theorem~\ref{thm:Main_Result_Single_Src_Dist} for the case 
$T>n+\min[m,n]$.  The proof of
Theorem~\ref{thm:Main_Result_Single_Src_Extnd} is presented in
\S\ref{sec:Proof_Main_Result_Single_Src_Extnd}.

\begin{theorem}\label{thm:Main_Result_Single_Src_Extnd}
If  $T>n+\min[m,n]$, then the capacity of $\mathrm{Ch}_{\textsl{m}}$
for $q\ge q_0$ is given by
\begin{align}
C_{\textsl{m}} &= \sum_{l=0}^{i^*} \psi(n,l) \gaussnum{i^*}{l} q^{-ni^*} \log_2\left( \frac{\gaussnum{T}{l}}{\gaussnum{i^*}{l}} \right)  \nonumber\\
&= i^*(T-i^*)\log_2{q} -\1_{\{n\le m\}} (T-i^*)\frac{\log_2{q}}{q} + q^{-1} + o(q^{-1}),
\end{align}
where $\1_{\{\cdot\}}$ is the indicator function and $q_0$ is the minimum field size that satisfies the set of inequalities
\begin{equation*}
\frac{\epsilon_{q_0}(l)-\epsilon_{q_0}(i^*)}{(T-n-i^*)(i^*-l)} \le \log_2{q_0}, \quad \forall l: 0\le l\le (i^*-1),
\end{equation*}
and
\begin{equation*}
\frac{\epsilon_{q_0}(l)-\epsilon_{q_0}(i^*)}{i^*(l-i^*)} \le \log_2{q_0}, \quad \forall l: (i^*+1)\le l\le m,
\end{equation*}
where $i^*=\min[m,n]$ and 
\[
\epsilon_{q}(l) \triangleq \sum_{d_y=0}^{\min[n,l]} \psi(n,d_y) \gaussnum{l}{d_y} q^{-nl} \log_2\left( \frac{\gaussnum{T}{d_y}}{\gaussnum{i^*}{d_y}} \right) 
-\min[n,l](T-i^*).
\]
The capacity is  achieved by sending matrices $X$
such that their rows span different $i^*$-dimensional subspaces.

Moreover, asymptotically in $T$, we can show that $q_0^{n-m+1}\ge
5m^2$ is sufficient for the case $m\le n$ and $q_0\ge nT$ is
sufficient if $m>n$.
\end{theorem}

Theorems~\ref{thm:Main_Result_Single_Src}~and~\ref{thm:Main_Result_Single_Src_Dist} 
state that the capacity behaves
as $i^*(T-i^*) \log_2 q$, for sufficiently large $q$. However, numerical
simulations indicate a very fast convergence to this value as $q$
increases. Fig.~\ref{fig:NumericalCapacity_P2P} depicts the capacity
for small values of $q$, calculated using the Differential Evolution
toolbox for MATLAB \cite{PrSt-JGO97}. This shows that the result is
relevant at much lower field size than dictated by the formalism of
the statement of Theorems~\ref{thm:Main_Result_Single_Src}~and~\ref{thm:Main_Result_Single_Src_Dist}.

\begin{figure}[thb!]
\begin{center}
\includegraphics[width=5in]{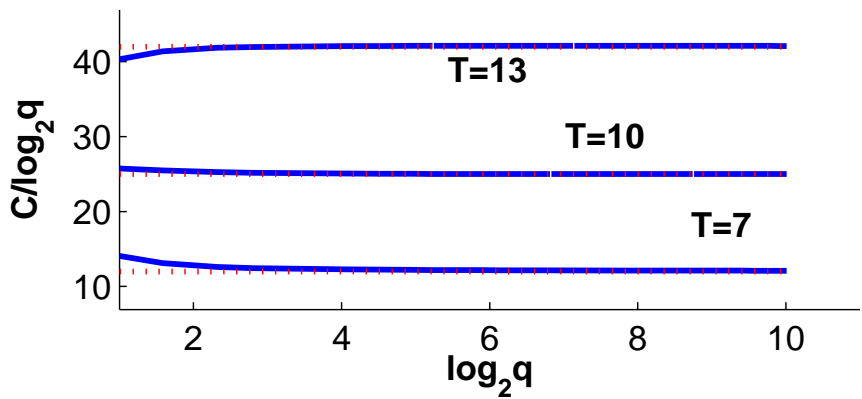}
\end{center}
\caption{Numerical calculation of the capacity for small values of $q$ and 
$m=11$, $n=7$. The dotted line depicts $i^*(T-i^*)$. }
\label{fig:NumericalCapacity_P2P}
\end{figure}

From Theorem \ref{thm:Main_Result_Single_Src_Dist}, we can derive the following 
guidelines for non-coherent network code design.
\subsubsection{Choice of subspaces}  
The optimal input distribution uses subspaces of a single dimension
equal to $\min[m,n]$ for $T\ge \min[m,n]+n$. As $T$ reduces, the set
of used subspaces gradually increases, by activating one by one
smaller and smaller dimensional subspaces, until, for $T\le n$, all
subspaces are used with equal probability\footnote{Note that although 
all the subspaces are equiprobable, we have distinct values for $\alpha_i^*$
since there are  different number of subspaces of each dimension.}. Fig.~\ref{fig:active}
pictorially depicts this gradual inclusion of subspaces. 

This behavior is different from the result of~\cite{SiKschKo-CapRandNetCod}
where  all the subspaces up to  dimension equal to the min-cut
appeared in the optimal input distribution. This difference is due to 
the different channel model used in our work and in \cite{SiKschKo-CapRandNetCod}.

\begin{figure}[thb!]
\begin{center}
\includegraphics{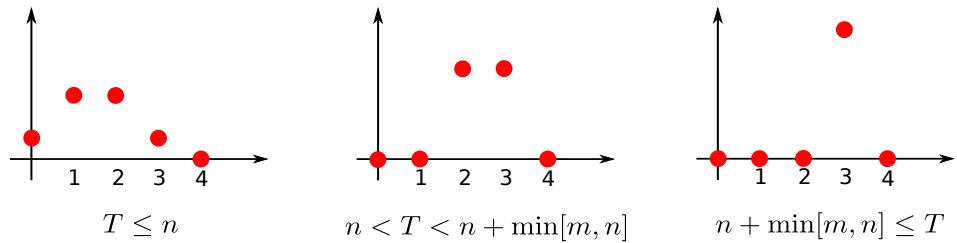}
\end{center}
\caption{Probability mass function of the active subspace dimensions 
for channel parameters $m=4$, $n=3$. As it is shown in Theorem~\ref{thm:Main_Result_Single_Src_Dist}
there exist three different regimes.}
\label{fig:active}
\end{figure}

\subsubsection{Values of m and n}
For a given and fixed packet length $T$, the optimal value of $m$ and $n$ 
equals $m=n=\lfloor T/2 \rfloor$ (optimality is in the sense of minimum 
requirement in order to obtain the maximum capacity for this $T$). For 
fixed $T$ and $m$, the optimal value of $n$ equals $n=\min[m,\lfloor T/2 \rfloor ]$. 
For fixed $T$ and $n$, the optimal value of $m$ equals $m=\min[n,\lfloor T/2 \rfloor ]$.

\begin{table}[tbh!] 
\caption{\label{tabl_1} Information loss from using coding vectors when $n= m$.} 
\begin{center} 
\begin{tabular}{|c|c|c|}\hline 
 & $T \le 2m$ & $T > 2m $  \\\hline 
$C_{\textsl{m}}-R_{\textsl{cv}}$ & $o(1)$ & $o(1)=(i^*-1)(T-i^*)\frac{\log_2{q}}{q} + O(q^{-1})$  \\\hline 
\end{tabular} 
\end{center} 
\end{table} 
\subsubsection{Subspace coding vs. coding vectors}  
One of the aims of this work was to find the regimes in which the
using of coding vectors \cite{ChWuJa-Alert03} is far from optimal.
Table~\ref{tabl_1} summarizes this difference. As we see
from the Table~\ref{tabl_1} subspace coding does not offer benefits as compared 
to the coding vectors approach for large field size\footnote{In the algebraic framework of \cite{KoetKsch-IT08-erasure},
the lifting construction used coding vectors, and they showed that this construction
achieves almost the same rates as optimal algebraic subspace codes. However, we
demonstrate in this paper that this phenomenon occurs for longer packet lengths
using an information-theoretic framework.}.

Table~\ref{tabl_1} is calculated as follows. 
The achievable rate $R_{\textsl{cv}}$ using coding vectors equals 
\[
R_{\textsl{cv}}\triangleq \Prob{\mathrm{rank}(G_k)=k}k(T-k)\log_2{q},
\] 
where $0<k\le m$ is the number of packets in each generation, \ie, each 
packet includes a coding vector of length $k$ and $T-k$ information 
symbols. Equivalently, we assume that we use $k$ of the $m$ possible input packets. The matrix $G_k$ is the $k\times k$ sub-matrix of $G$ that is applied
over the input packets. To calculate $R_{\textsl{cv}}$, 
we know that $\Prob{\mathrm{rank}(G_k)=k}=\prod_{i=0}^{k-1} (1-q^{-k+i})=1-q^{-1}+O(q^{-2})$. Assume we choose $k=i^*$ we have 
$R_{\textsl{cv}}= i^*(T-i^*)\log_2{q} - i^*(T-i^*)\frac{\log_2{q}}{q}$,
where $i^*=\min\left[m,n,\lfloor T/2\rfloor \right]$. For the capacity
$C_{\textsl{m}}$ we use the large $q$-regime as considered in Theorem~\ref{thm:Main_Result_Single_Src} for the
case $T\le 2m$ and the finite $q$-regime of Theorem~\ref{thm:Main_Result_Single_Src_Extnd} for the
case $T>2m$.

\subsection{Extension to the packet erasure networks}
After the error free single source scenario, we consider  packet erasure
networks, and calculate an upper and lower bound on the capacity for
this case. The work in \cite{SiKschKo-CapRandNetCod}, which is the closest to ours,
did not consider erasures but instead constant-dimension additive errors. 
In practice, depending on the application, either of the models might be more suitable: for example, if network coding is deployed at an application layer, then, unless there exist malicious attackers, packet erasures are typically used to abstract both the underlying physical channel errors, as well as packet dropped at queues or lost due to expired timers.

We model the erasures in the network as an end-to-end
phenomenon which randomly erases packets according to some probability
distribution. 
Formally, we rewrite the channel defined in
\eqref{eq:channel_model_P2P} as
\begin{equation}\label{eq:channel_model_P2P_erasure}
Y[t]=E[t]G[t]X[t],
\end{equation}
where $G\in\Fbb_q^{m\times m}$ is assumed to be a squre chanel matrix and $E\in\Fbb_q^{m\times m}$ is a diagonal random matrix whose elements 
on its diagonal are either $1$ or $0$. We also assume that $q$ is large, and as a result the
transfer matrix is full rank with high probability. 
Moreover, we consider the case where $m\leq \frac{T}{2}$, i.e. the matrix $X$ is a fat matrix.
Recall that we can think of the rows of this matrix as packets send by the source, and  the rows of the $Y$ matrix as packets received at the destination.

Note that in equation (\ref{eq:channel_model_P2P_erasure})
all of the erasure events are captured by the erasure matrix $E[t]$.
Moreover,  the erasure pattern is important only up to determining the number of packets that the destination receives, 
since the transfer matrix $G[t]$ 
is unknown and  distributed uniformly at random over all full rank matrices.
Thus, we model the number of received packets (number of non-zero elements 
on the diagonal of $E[t]$) as a random variable $N$ which takes values in $0\le N\le m$ 
according to some distribution that depends on the packet erasures in the 
network.  
In this case the capacity is
\[
C_e = \max_{P_X} I(X;Y,N).
\]
We can then use our previous result,
Theorem~\ref{thm:Main_Result_Single_Src}, to find an upper and lower
bound for the capacity $C_e$ when we have packet erasure in the
network, as the following
Theorem~\ref{thm:P2P_erasure} describes.

\begin{theorem}\label{thm:P2P_erasure}
Let the number of received packets at the destination 
be a random variable $N$ defined 
over the set of integers $0\le N\le m$. Also, assume that $m\le  T/2$. 
Then for large $q$, we have the following upper and lower bound for the capacity $C_e$,
\begin{align*}
\mu_1(T-m)\log_2{q} \le C_e \le \mu_1\left(T-\frac{\mu_2}{\mu_1}\right)\log_2{q},
\end{align*}
where $\mu_1\triangleq\Expc{N}{N}$ and $\mu_2\triangleq\Expc{N}{N^2}$.
\end{theorem}
For the proof of Theorem~\ref{thm:P2P_erasure} and more discussion 
refer to Appendix~\ref{sec:apndx2}.

{Note that because we do not necessarily employ full-rank matrices $X$,
it is possible that although some packets are erased at the destination,
the received packets still span a matrix of the same rank as  $X$; thus erasing packets is not equivalent to erasing dimensions.}

\subsection{Multiple Sources}

In several practical applications, such as sensor networks, data sources 
are not necessarily co-located. We thus extend our work to the case where 
multiple not co-located sources transmit information to a common receiver. 
In particular, we consider the non-coherent MAC  introduced in 
Definition~\ref{def:matrix_channel_MAC}, and characterize the capacity region 
of this network for the case of two sources with $m_1$ and $m_2$ input packets 
and packet length $T> 2(m_1+m_2)$. 
We believe that this technique can be extended to more than two sources.

To find the rate region of the matrix multiple access channel
$\mathrm{Ch}_{\textsl{m-MAC}}$, we first show that the two channels
$\mathrm{Ch}_{\textsl{m-MAC}}$ and $\mathrm{Ch}_{\textsl{s-MAC}}$ 
are equivalent, as stated in Theorem~\ref{thm:channel_equivalence_MAC}.
We then find the rate region of the subspace multiple access channel
$\mathrm{Ch}_{\textsl{s-MAC}}$ which is stated in Theorem~\ref{thm:Main_Result_Multpl_Src}.
To avoid repetition, we state Theorem~\ref{thm:channel_equivalence_MAC} 
without a proof because its proof is very similar to
that of Theorem~\ref{thm:channel_equivalence_P2P}.

\begin{theorem}\label{thm:channel_equivalence_MAC}
The matrix MAC  $\mathrm{Ch}_{\textsl{m-MAC}}$ defined in 
Definition~\ref{def:matrix_channel_MAC} is equivalent to 
the subspace MAC  $\mathrm{Ch}_{\textsl{s-MAC}}$ defined in 
Definition~\ref{def:subspace_channel_MAC} in the sense that the 
optimal rate region for these two channels is the same.
\end{theorem}

\begin{theorem}\label{thm:Main_Result_Multpl_Src}
For $\frac{T}{2}> m_1+m_2$, the asymptotic (in the field size $q$) capacity region of the MAC  $\mathrm{Ch}_{\textsl{m-MAC}}$ 
introduced in Definition~\ref{def:matrix_channel_MAC} 
is given by
\begin{align}\nonumber
\mathcal{R}^* \triangleq \mathrm{convex~hull } \bigcup_{(d_1,d_2)\in\mathcal{D}^* } \mathcal{R}(d_1,d_2),
\end{align}
where
\begin{align}
\mathcal{R}(d_1,d_2)\triangleq\{(R_1,R_2): R_i\leq R_i(d_1,d_2),\ i=1,2\},
\label{eq:def:Ri}
\end{align} 
\begin{align*}
R_i(d_1,d_2)\triangleq d_i(T-d_1-d_2) \log_2 q, \qquad i=1,2,
\end{align*}
and
\begin{align}\nonumber
\mathcal{D}^*\triangleq\{(d_1,d_2): \;\; & 0\leq d_i \leq \min[n,m_i],
 0\leq d_1+d_2 \leq \min[n, m_1+m_2]\}.
\end{align}
\end{theorem}
We note that the rate region forms a polytopes that has the 
following number of corner points (see Corollary~\ref{cor:NumCornerPointsMAC} 
in  \S\ref{sec:ProofMultiSource})
\[
\min\left[m_1, (n-m_2)^+\right] + \min\left[m_2, (n-m_1)^+\right] + 2 - \1_{\{n\ge m_1+m_2\}}.
\]
The rate region $\mathcal{R}^*$ is shown in Fig.~\ref{fig:MAC_region} 
for a particular choice of parameters.

The proof of this theorem is provided in \S\ref{sec:ProofMultiSource}. 
We first derive an outer bound by deriving two other bounds: a cooperative bound and a 
coloring bound. For the coloring bound, we utilize a combinatorial 
approach to bound the number of \emph{distinguishable} symbol pairs 
that can be transmitted from the sources to the receiver.
 We then show that a simple scheme that uses coding vectors achieves 
the outer bound. We thus conclude that, for the case of two sources 
when  $\frac{T}{2}> m_1+m_2$, use of coding vectors is (asymptotically) 
optimal.

\begin{figure}
\begin{center}
	\includegraphics[width=4.5in]{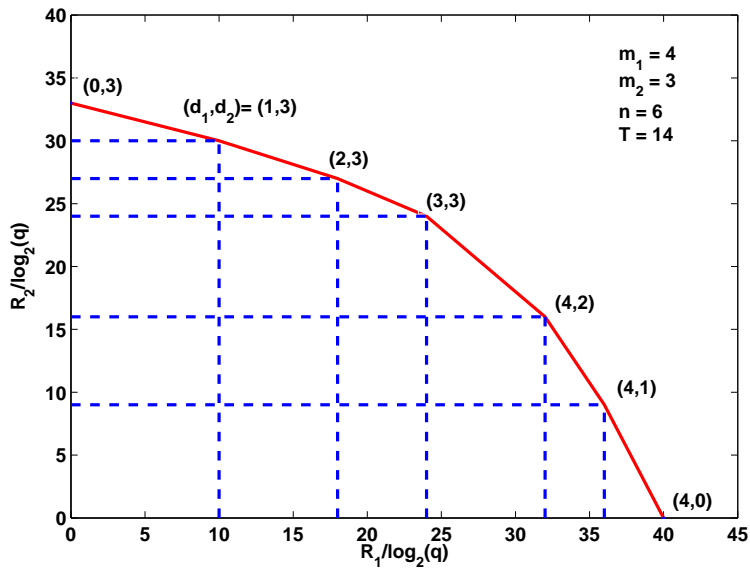}
\end{center}
\caption{The MAC region $\mathcal{R}^*$ for parameters $m_1=4$, $m_2=3$, $n=3$, $T=14$. }
\label{fig:MAC_region}
\end{figure}

\section{The Channel Capacity: Single Source Scenario}\label{sec:TheChanlCap-SnglSrc}
In this section we will prove Theorem~\ref{thm:Main_Result_Single_Src},
Theorem~\ref{thm:Main_Result_Single_Src_Dist}, and 
Theorem~\ref{thm:Main_Result_Single_Src_Extnd}. 

\subsection{Equivalence of the Matrix Channel $\mathrm{Ch}_{\textsl{m}}$ 
and the Subspace Channel $\mathrm{Ch}_{\textsl{s}}$}\label{subsec:Equiv-MatrixChannel-SubspaceChannel}
For convenience let us rewrite the channel \eqref{eq:channel_model_P2P} again\footnote{In 
the rest of the paper we will omit for convenience the time index $t$.}
\[
Y=GX.
\]
To find the capacity of the above channel we need to maximize the mutual 
information between the input and the output of the channel with respect 
to the input distribution $P_X(x)$. Since the rows of $G$ are chosen 
independently of each other, assuming that a matrix $X=x$ has been transmitted,  we can think of the 
rows of the received matrix $Y$ as chosen independently from each other, among all 
the possible vectors in the row span of $x$. The independence of rows of $Y$ allows us 
to write the conditional probability of $Y$ given $X$, referred to as the channel transition 
probability, as follows
\begin{equation}\label{eq:P2P_channel_transfer_prob_1}
P_{Y|X}(y|x)=\left\{ \begin{array}{ll} q^{-n\dim(\sspan{x})} &
  \sspan{y}\sqsubseteq \sspan{x},\\ 0 & \text{otherwise},
\end{array} \right.
\end{equation}
where $x\in\mathcal{X}=\Fbb_q^{m\times T}$, and $y\in\mathcal{Y}=\Fbb_q^{n\times T}$.

The mutual information $I(X;Y)$ between $X$ and $Y$ is a
function of $P_X(x)$ and $P_{Y|X}(y|x)$ that can be expressed as
\begin{equation}
I(X;Y)= \sum_{\begin{subarray}{l} x\in\mathcal{X},\\ y\in\mathcal{Y}\end{subarray}}  P_X(x) P_{Y|X}(y|x)\log_2\left(\frac{P_{Y|X}(y|x)}{P_Y(y)}\right).
\label{eq:I_1}
\end{equation}
It is clear from (\ref{eq:P2P_channel_transfer_prob_1}) that 
$P_{Y|X}(y|x_1)=P_{Y|X}(y|x_2)$ for all $x_1,x_2\in\mathcal{X}$ such 
that $\sspan{x_1}=\sspan{x_2}$ which reveals symmetry for the 
channel $\mathrm{Ch}_{\textsl{m}}$. We exploit this symmetry
to show that $C_{\textsl{m}}=C_{\textsl{s}}$ as it is stated in 
Theorem~\ref{thm:channel_equivalence_P2P} and proved in 
Appendix~\ref{sec:apndx1}.

The proof of Theorem~\ref{thm:channel_equivalence_P2P} determines
how we can map an input distribution of $\mathrm{Ch}_{\textsl{s}}$
to an input distribution for $\mathrm{Ch}_{\textsl{m}}$ that achieves the same mutual information.
The input distribution $P_X(x)$ should be chosen such that we have
$\sum_{x\in\Xca:\sspan{x}=\pi_x} P_X(x)=P_{\Pi_X}(\pi_x)$. One simple
way to do this is to put all the probability mass of $\pi_x$ on one
matrix $x$ such that $\sspan{x}=\pi_x$.

\subsection{Upper and Lower bound for the Capacity of $\mathrm{Ch}_{\textsl{m}}$}\label{subsec:Capacity-UpperLowerBound}
Here, we state the proof of Theorem~\ref{thm:Main_Result_Single_Src} 
by giving upper and lower bounds for the capacity that differ in
$o(1)$ bits, which vanishes  as $q\rightarrow\infty$.

Let $C_{\textsl{m}}(n,m,T)$ denote the capacity of the channel 
$\mathrm{Ch}_{\textsl{m}}$. Let $C_{\textsl{f-m}}(n,m,T)$ denote the
capacity of the channel $Y=AX$ where $A\in\Fbb_q^{n\times m}$ is a 
full-rank matrix chosen uniformly at random among all the full-rank matrices 
in $\Fbb_q^{n\times m}$. Then, we have the following lemma.

\begin{lemma}\label{lem:Cap_SingleSrc_UpperLower}
We can bound $C_{\textsl{m}}(n,m,T)$ from above and below as follows
\[
C_{\textsl{m}}(h,h,T) \le C_{\textsl{m}}(n,m,T) \le C_{\textsl{f-m}}(n,m,T) \le C_{\textsl{f-m}}(h,h,T),
\]
where $h=\min[m,n]$. 
\end{lemma}

\begin{proof}
Let $U_{n\times m}\in\Fbb_q^{n\times m}$ denote a generic 
random matrix chosen uniformly at random and independently from any other variable. Similarly, let 
$A_{n\times m}\in\Fbb_q^{n\times m}$ denote a generic \emph{full-rank} matrix 
chosen uniformly at random among all such full-rank matrices and 
independent from any other variable. (Note that each new instance of 
such a matrix in the same equation denotes a different random variable
which is independent from the other random variables.)

Since the channel $Y=A_{n\times m}X$ is statistically equivalent to the
channel $Y=A_{n\times n}A_{n\times m}A_{m\times m}X$, we have, by the data 
processing inequality, that $C_{\textsl{f-m}}(n,m,T) \le C_{\textsl{f-m}}(h,h,T)$.

Using the same argument, since the channel $Y=U_{n\times m}X$ is equivalent 
to the channel $Y=U_{n\times n}A_{n\times m}X$ if $n\ge m$, and is equivalent 
to the channel $Y=A_{n\times m}U_{m\times m}X$ if $n\le m$ we have
$C_{\textsl{m}}(n,m,T) \le C_{\textsl{f-m}}(n,m,T)$.

To obtain the lower bound we proceed as follows. Let us choose 
$X=\left[\begin{smallmatrix} I_h\\ \\ 0\end{smallmatrix}\right]\overline{X}$
and $\overline{Y}=\left[I_h\ \ 0\right] Y$, where $Y=U_{n\times m}X$. Then we can write
\[
\overline{Y}=\left[I_h\ \ 0\right] U_{n\times m} {I_h \brack 0}\overline{X} = U_{h\times h}\overline{X},
\]
where $U_{h\times h}$ is the  upper left $h\times h$ sub-matirx of $U_{n\times m}$.
Thus, again the data processing inequality implies that  $C_{\textsl{m}}(h,h,T) \le C_{\textsl{m}}(n,m,T)$.
\end{proof}

\begin{lemma}\label{lem:Cap_SingleSrc_Upper}
For $C_{\textsl{m}}(n,m,T)$ we have
\[
C_{\textsl{m}}(n,m,T) \le i^*(T-i^*) \log_2{q} + o(1),
\]
where $i^*=\min[m,n,\lfloor T/2\rfloor]$.
\end{lemma}

\begin{proof}
By Lemma~\ref{lem:Cap_SingleSrc_UpperLower} we have 
\begin{align*}
C_{\textsl{m}}(n,m,T) &\le C_{\textsl{f-m}}(h,h,T) \\
&\stackrel{(a)}{=} \log_2\left(\sum_{i=0}^h \gaussnum{T}{i} \right)\\
&\stackrel{(b)}{=} i^*(T-i^*) \log_2{q} +o(1),  
\end{align*}
where $(a)$ follows from \cite[Corollary~2]{SiKschKo-CapRandNetCod} and 
$(b)$ follows from Lemma~\ref{lem:gauss_Approximation}.
\end{proof}

\begin{lemma}\label{lem:Cap_SingleSrc_Lower}
For $C_{\textsl{m}}(n,m,T)$ we have
\[
C_{\textsl{m}}(n,m,T) \ge i^*(T-i^*) \log_2{q} - o(1),
\]
where $i^*=\min[m,n,\lfloor T/2\rfloor]$.
\end{lemma}

\begin{proof}
For every subspace $\Pi \in \mathrm{Gr}(T,i^*)$, let $\mathrm{RREF}(\Pi) \in \Fbb_q^{i^*\times T}$
be a matrix in reduced row echelon form such that $\Pi=\sspan{\mathrm{RREF}(\Pi)}$. 
Choose $X=\left[\begin{smallmatrix} I_{i^*} \\ \\ 0 \end{smallmatrix}\right]\times \mathrm{RREF}(\Pi_X)\in\Fbb_q^{m\times T}$,
where $\Pi_X$ is chosen uniformly at random from $\mathrm{Gr}(T,i^*)$. 
Define the random variable $Q = \1_{\{\mathrm{rank}(Y)=i^*\}}$. Note that
$\Pi_Y=\Pi_X$ when $Q = 1$. Thus, we have $H(\Pi_Y|\Pi_X,Q = 1) = 0$ and 
$H(\Pi_Y|Q = 1) = H(\Pi_X) = \log_2\gaussnum{T}{i^*}\ge i^*(T-i^*)\log_2{q}$.
Then, it follows that
\begin{align*}
C_{\textsl{m}}(n,m,T) &\stackrel{(a)}{\ge} 
C_{\textsl{m}}(h,h,T) \\
&\stackrel{(b)}{\ge} I(\Pi_X;\Pi_Y)\\
&\stackrel{(c)}{=} I(\Pi_X;\Pi_Y,Q)\\
&= I(\Pi_X;Q) + I(\Pi_X;\Pi_Y|Q)\\
&\ge \Prob{Q=1}I(\Pi_X;\Pi_Y|Q=1)\\
&\ge \Prob{Q=1}i^*(T-i^*)\log_2{q},
\end{align*}
where $(a)$ is due to Lemma~\ref{lem:Cap_SingleSrc_UpperLower}, $(b)$ follows follows from Theorem~\ref{thm:channel_equivalence_P2P}, and $(c)$ 
holds since $Q$ is a deterministic function of $\Pi_Y$.
Now, note that we can write
\begin{align*}
\Prob{Q=1} &= \Prob{\mathrm{rank}(U_{h\times h}X)=i^*}\\
&= \Prob{\mathrm{rank}\left(U_{h\times h} \left[\begin{smallmatrix} I_{i^*} \\ \\ 0 \end{smallmatrix}\right] \right)=i^*}\\
&= \Prob{\mathrm{rank}(U_{h\times i^*})=i^*}\\
&\ge 1-\frac{i^*}{q^{k-i^*+1}}\\
&\ge 1-\frac{i^*}{q},
\end{align*}
and thus we obtain the desired result.
\end{proof}

Combining Lemma~\ref{lem:Cap_SingleSrc_Upper} and Lemma~\ref{lem:Cap_SingleSrc_Lower}
recovers Theorem~\ref{thm:Main_Result_Single_Src}.

\subsection{The Optimal Solution: General Approach}\label{subsec:OptSolution_SingleSrc_GenApproach}
Generally, we are interested in finding the capacity and input distribution of
$\mathrm{Ch}_{\textsl{m}}$ exactly. It is shown in Theorem~\ref{thm:channel_equivalence_P2P}
that instead of the channel $\mathrm{Ch}_{\textsl{m}}$ we can focus on the
channel $\mathrm{Ch}_{\textsl{s}}$. Thus, we are interested in optimizing 
the following quantity
\begin{equation}\label{eq:mutual_information_s_1}
I(\Pi_X;\Pi_Y)= \sum_{\begin{subarray}{l} \pi_x\in\wt{\mathcal{X}},\\ \pi_y\in\wt{\mathcal{Y}}\end{subarray}}  P_{\Pi_X}(\pi_x) P_{\Pi_Y|\Pi_X}(\pi_y|\pi_x)\log_2\left(\frac{P_{\Pi_Y|\Pi_X}(\pi_y|\pi_x)}{P_{\Pi_Y}(\pi_y)}\right).
\end{equation}
Remember that $\wt{\Xca}=\mathrm{Sp}(T,m)$ and $\wt{\Yca}=\mathrm{Sp}(T,n)$.

The following lemma states that the optimal solution for the channel 
$\mathrm{Ch}_{\textsl{s}}$ should be uniform over all subspaces with the 
same dimension, as it is intuitively expected from the symmetry of the channel.
\begin{lemma}\label{lem:uniform_dist}
The input distribution that maximizes $I(\Pi_X;\Pi_Y)$ for $\mathrm{Ch}_{\textsl{s}}$ 
is the one which is uniform over all subspaces having the same dimension.
\end{lemma}

Lemma \ref{lem:uniform_dist} shows that the optimal input distribution can be expressed as
\begin{equation}\label{eq:optimal_dist_form}
\Prob{\Pi_X=\pi_x} = \frac{\alpha_{d_x}}{\gaussnum{T}{d_x}},
\end{equation}
where $d_x=\dim(\pi_x)$, $\alpha_{d_x}=\Prob{\dim(\Pi_X)=d_x}$, and we have 
$\sum_{d_x=0}^{\min[m,T]} \alpha_{d_x}=1$.
We can then simplify $I(\Pi_X;\Pi_Y)$ as stated in the following lemma.
\begin{lemma}\label{lem:P2P_MutualInfo_Subspace_Final}
Assuming an optimal input probability distribution of the form in (\ref{eq:optimal_dist_form}), 
the mutual information $I(\Pi_X;\Pi_Y)$ can be simplified to
\begin{align}\label{eq:P2P_MutualInfo_Subspace_Final}
I(\Pi_X;\Pi_Y) =& -\sum_{d_x=0}^{\min[m,T]} \alpha_{d_x} nd_x  \log_2{q} \nonumber\\
& - \sum_{d_x=0}^{\min[m,T]} \alpha_{d_x} q^{-nd_x} \sum_{d_y=0}^{\min[n,d_x]} \psi(n,d_y)\gaussnum{d_x}{d_y} \log_2(f(d_y)),
\end{align}
where
\begin{equation}\label{eq:f(d_y)}
f(d_y)\triangleq\frac{P_{\Pi_Y}(\pi_y)}{\psi(n,d_y)}=\frac{1}{\gaussnum{T}{d_y}} \sum_{d_x=d_y}^{\min[m,T]} \gaussnum{d_x}{d_y} q^{-nd_x} \alpha_{d_x}.
\end{equation}
\end{lemma}

Lemmas~\ref{lem:uniform_dist}~and~\ref{lem:P2P_MutualInfo_Subspace_Final} 
show that the problem of finding 
the optimal input distribution for the channel $\mathrm{Ch}_{\textsl{s}}$ 
is reduced to finding the optimal choice for $\alpha_i,\ i=0,\dots,\min[m,T]$. 
We know that the mutual information is a concave function with respect to 
$P_{\Pi_X}(\pi_x)$'s. Observation~\ref{obs:Lin_Trans_Preserve_Concavity} implies 
that because \eqref{eq:optimal_dist_form} is a linear transformation from  
$P_{\Pi_X}(\pi_x)$'s to $\alpha_i$'s, as a result the mutual information $I(\Pi_X;\Pi_Y)$ 
is also concave with respect to $\alpha_i$'s \cite{BoVa-ConvexOpt}.

\begin{observation}\label{obs:Lin_Trans_Preserve_Concavity}
Let $g(\mathbf{x})$ be a concave function and let $\mathbf{x}=h(\mathbf{z})$ be 
a linear transform from $\mathbf{z}$ to $\mathbf{x}$. Then $g(h(\mathbf{z}))$ 
is also a concave function.
\end{observation}

Using Observation~\ref{obs:Lin_Trans_Preserve_Concavity}, we know that the mutual 
information is a concave function with respect to 
$\alpha_i$'s. This allows us to use the Kuhn-Tucker theorem \cite{BoVa-ConvexOpt} 
to solve the convex optimization problem. According to this theorem, the set 
of probabilities $\alpha_i^*$, $0\le i\le\min[m,T]$, maximize the mutual 
information if and only if there exists some constant $\lambda$ such that
\begin{equation}\label{eq:KuhnTuckerCond_I}
\left\{\begin{array}{ll} 
\left. \frac{\partial I(\Pi_X;\Pi_Y)}{\partial \alpha_k} \right|_{\boldsymbol{\alpha}^*}=\lambda & \forall k:\  \alpha^*_k>0,\\
\\
\left. \frac{\partial I(\Pi_X;\Pi_Y)}{\partial \alpha_k}\right|_{\boldsymbol{\alpha}^*} \le\lambda & \forall k:\  \alpha^*_k=0,
\end{array}\right.
\end{equation}
where $\sum_{i=0}^{\min[m,T]} \alpha^*_i=1$, $0\le k\le\min[m,T]$, and $\boldsymbol{\alpha}^*$ is the vector of the optimum input probabilities of choosing subspaces of certain dimension,
\[
\boldsymbol{\alpha}^* = \left[\begin{array}{ccc} \alpha^*_0 & \cdots & \alpha^*_{\min[m,T]} \end{array}\right]^\mathrm{T}.
\]

\begin{lemma}\label{lem:I_Derivative_Simplified}
By taking the partial derivative of the mutual information given in 
\eqref{eq:P2P_MutualInfo_Subspace_Final} with respect to $\alpha_k$, we have
\begin{equation}\label{eq:I_Derivative_Simplified}
I'_k \triangleq \frac{\partial I(\Pi_X;\Pi_Y)}{\partial \alpha_k} = -nk\log_2{q} - \sum_{d_y=0}^{\min[n,k]} \psi(n,d_y) \gaussnum{k}{d_y} q^{-nk} \log_2\left(f(d_y) \right) -\log_2{e}.
\end{equation}
\end{lemma}

Multiplying both sides of (\ref{eq:I_Derivative_Simplified}) by $\alpha_k$ and summing over $k$ we get
\begin{align*}
I - \log_2{e} = \sum_{k=0}^{\min[m,T]} \alpha_k I'_k.
\end{align*}
By choosing the optimal values $\alpha_k=\alpha^*_k$ for $0\le k\le\min[m,T]$, the RHS becomes $\lambda$, and the mutual information increases to $C_{\textsl{s}}$. So we may write
$\lambda = C_{\textsl{s}} - \log_2{e}$.

\subsection{Solution for Large Field Size}\label{subsec:OptSolution_SingleSrc_Large_q}
In this subsection, we focus on large size fields, $q\gg 1$. This 
assumption allows us to use some approximations to simplify the 
conditions in \eqref{eq:KuhnTuckerCond_I}. 
Assuming large $q$ we can rewrite \eqref{eq:I_Derivative_Simplified} as follows
\begin{align}\label{eq:I_Derivative_midl1}
I'_k = -nk\log_2{q} -\log_2{e} - \sum_{d_y=0}^{\min[n,k]} \left(1+O(q^{-1})\right) q^{-(n-d_y)(k-d_y)} \log_2\left(f(d_y) \right), 
\end{align}
where we have used Lemma~\ref{lem:gauss_Approximation} and 
Lemma~\ref{lem:psi_Approximation}. Using similar approximations, $\log_2 f(d_y)$ 
defined in \eqref{eq:f(d_y)} can be approximated as
\begin{align}\label{eq:f_Order_Approx}
\log_2\left(f(d_y)\right) =& -d_yT\log_2{q} + O(q^{-1}) + \log_2\left(\sum_{d_x=d_y}^{\min[m,T]} q^{-(n-d_y)d_x} \alpha_{d_x} \right). 
\end{align}
Then we have the following result, Lemma~\ref{lem:DominTerm_DerivativeI}.

\begin{lemma}\label{lem:DominTerm_DerivativeI}
The dominating term in the summation in (\ref{eq:I_Derivative_midl1}) is the one obtained for $d_y=\min[n,k]$.
\end{lemma}

From the proof of Lemma~\ref{lem:DominTerm_DerivativeI} 
written in Appendix~\ref{sec:apndx1}, we can also see that the remaining terms in 
the summation of \eqref{eq:I_Derivative_midl1} are of order $o(1)$, so we can 
write
\begin{align}\label{eq:I_Derivative_Approx}
I'_k =& [T\min[n,k]-nk]\log_2{q} + \underbrace{o(1)}_{\epsilon_q(k)} -\log_2{e} -\log_2\left(\sum_{d_x=\min[n,k]}^{\min[m,T]} q^{-[n-\min[n,k]]d_x} \alpha_{d_x} \right).
\end{align}

Assuming that the expression inside the $\log(\cdot)$ function in (\ref{eq:I_Derivative_Approx}) is not zero for every $0\le k\le \min[m,T]$, we can rewrite the Kuhn-Tucker conditions as
\begin{equation*}
\sum_{d_x=\min[n,k]}^{\min[m,T]} q^{-[n-\min[n,k]]d_x} \alpha_{d_x} \ge 
2^{-C_{\textsl{s}}+o(1)} q^{[T\min[n,k]-nk]},
\end{equation*}
where the inequality holds with equality for all $k$ with $\alpha_k^*>0$.

Let $\delta\triangleq \min[m,T]$ and define the $(\delta+1)\times(\delta+1)$ matrix $\mathbf{A}$ with elements 
\begin{align*}
\mathbf{A}_{ij}\triangleq \left\{\begin{array}{ll}
q^{-[n-\min[n,i]]j}  & \min[n,i] \leq j \leq \delta,\\
0 & \textrm{otherwise}.
\end{array}
\right.
\end{align*}

We also define the column vector $\mathbf{b}$ with elements $\mathbf{b}_i\triangleq q^{[T\min[n,i]-ni]}$ for $0\leq i \leq \delta$. Note that for convenience the indices of matrix $\mathbf{A}$ and vector $\mathbf{b}$ start from $0$. Using these definitions, we are able to rewrite the Kuhn-Tucker conditions in the matrix form as 
\begin{equation}\label{eq:KuhnTuckerCond_MatrixForm}
\mathbf{A}\boldsymbol{\alpha}^*  \succeq 2^{-C_{\textsl{s}}+o(1)} \mathbf{b}.
\end{equation}
In the following, we consider two cases for $\delta\le n$ and $\delta>n$, and find $\boldsymbol{\alpha}^*$ for each of them, separately. 

\noindent\textbf{First case: $\delta\le n$.}
In this case  we can explicitly write the matrix $\mathbf{A}$ and vector $\mathbf{b}$ as 
\begin{equation*}
\mathbf{A}=\left[\begin{array}{cccccc} 
1 & q^{-n}  & \cdots & q^{-(\delta-1)n} & q^{-\delta n}\\
0 & q^{-(n-1)} &  \cdots & q^{-(\delta-1)(n-1)} & q^{-\delta (n-1)}\\
0 & 0 &  \cdots & q^{-(\delta-1)(n-2)} & q^{-\delta (n-2)}\\
\vdots & \vdots & \ddots & \vdots & \vdots\\
0 & 0 & \cdots & q^{-(\delta-1)(n-\delta+1)} & q^{-\delta(n-\delta+1)}\\
0 & 0 & \cdots & 0 & q^{-\delta(n-\delta)}
\end{array} \right],
\end{equation*}
and
\begin{equation*}
\mathbf{b}=\left[\begin{array}{cccc} 1 & q^{(T-n)} & \cdots & q^{\delta(T-n)}  \end{array}\right]^\mathrm{T}.
\end{equation*}

The fact that the expression inside the $\log_2(\cdot)$ function in (\ref{eq:I_Derivative_Approx})  is non-zero for $k=\delta$, forces $\alpha_{\delta}^*$ to be positive. Thus the last row of the matrix inequality in \eqref{eq:KuhnTuckerCond_MatrixForm} should be satisfied as an equality. Therefore,
\begin{align*}
\alpha^*_{\delta} =& \frac{q^{\delta(T-n)}}{q^{-\delta(n-\delta)}} 2^{-C_{\textsl{s}}+o(1)} 
=  q^{\delta(T-\delta)} 2^{-C_{\textsl{s}}+o(1)}.
\end{align*}

Now we use induction to show that the optimal solution has the form
\begin{equation}
\alpha^*_i= \left\{\begin{array}{lcl}
q^{i(T-i)} 2^{-C_{\textsl{s}}+o(1)} & : & \kappa\le i\le\delta,\\
0 & : & 0\le i<\kappa,
\end{array}\right.
\label{eq:optimal_case1}
\end{equation}
where we will determine $\kappa$ later. 

Let us fix $l$ and assume that $\alpha^*_i=q^{i(T-i)} 2^{-C_{\textsl{s}}+o(1)}$ 
for $0\le l<i\le\delta$. Then for $\alpha^*_{l}$ we can write
\begin{equation*}
A_{ll}\alpha^*_l + \sum_{j=l+1}^{\delta} q^{-(n-l)j} \alpha^*_j \ge q^{l(T-n)} 2^{-C_{\textsl{s}}+o(1)},
\end{equation*}
or equivalently
\begin{align}\label{eq:KuhnTucker_Cond_dim_l}
A_{ll}\alpha^*_l \ge & q^{l(T-n)} 2^{-C_{\textsl{s}}+o(1)}  - \sum_{j=l+1}^{\delta} q^{-(n-l)j} \alpha^*_j \nonumber\\
=& q^{l(T-n)} 2^{-C_{\textsl{s}}+o(1)} \left[1 - \sum_{j=l+1}^{\delta} q^{(T-n-j)(j-l)}\right].
\end{align}
We can use induction for one step more to show that $\alpha^*_l$ 
is of the desired form \eqref{eq:optimal_case1} if the previous 
expression is satisfied with equality. This is true if we have
$1 - \sum_{j=l+1}^{\delta} q^{(T-n-j)(j-l)} \ge 0$,
or equivalently (assuming large $q$) if we have
$\left.(T-n-j)\right|_{j=l+1} <0$.
So we can conclude that we should have
$(T-n)^+\le l \le \delta$.
It can be easily verified that for $i<(T-n)^+$ the Kuhn-Tucker 
equation for $\alpha^*_i$ satisfies the strict inequality so 
$\alpha^*_i=0$ for $i<\min[(T-n)^+,\delta]$. The above 
argument results in a solution of the following form for 
the case $\delta\le n$ 
\begin{equation}\label{eq:Solution_Form_big_n}
\alpha^*_i= \left\{\begin{array}{lcl} 
q^{i(T-i)} 2^{-C_{\textsl{s}}+o(1)} & : & \min\left[(T-n)^+,\delta\right]\le i\le\delta,\\
0 & : & 0\le i < \min\left[(T-n)^+,\delta\right].
\end{array}\right.
\end{equation}

\noindent\textbf{Second case: $\delta>n$.}
We now write matrix $\mathbf{A}$ and vector $\mathbf{b}$ as 
\begin{align*}
\mathbf{A}&=\left[\begin{matrix}
1 & q^{-n} & \cdots & \cdots & \cdots & \cdots & q^{-\delta n}\\
0 & q^{-(n-1)} & \cdots & \cdots & \cdots & \cdots & q^{-\delta (n-1)}\\
\vdots & \ddots & \ddots & \vdots & \vdots & \vdots & \vdots \\
0 & \cdots &  0 &  q^{-(n-1)} & q^{-n} & \cdots & q^{-\delta}\\
\hline
0 &  \cdots & 0 & 0 & 1 & \cdots & 1\\
\vdots  & \ddots & \vdots & \vdots & \vdots & \ddots & \vdots\\
0 &  \cdots & 0 & 0 & 1 & \cdots & 1\\
\end{matrix} \right],
\end{align*} 
and
\begin{equation*}
\mathbf{b}=\left[\begin{matrix} 1 & q^{(T-n)} & \cdots & q^{(n-1)(T-n)} & q^{n(T-n)} & q^{n(T-n-1)} & \cdots & q^{n(T-\delta)}  \end{matrix}\right]^\mathrm{T}.
\end{equation*}
The last $\delta-n+1$ rows of $\mathbf{A}$ are the same while 
$b_i$ is decreasing with $i$ for $i\ge n$. Thus, the last 
$\delta-n$ inequalities are strict and therefore,
\begin{align}
\alpha^*_{n+1}=\cdots=\alpha^*_{\delta}=0.
\label{eq:Solution_Form_small_n_1}
\end{align}

The remaining equations can simply be reduced to the first case. Define 
\begin{equation*}
\tilde{\mathbf{A}}=\left[\begin{array}{cccccc} 
1 & q^{-n}  & \cdots & q^{-(n-1)n} & q^{-n^2}\\
0 & q^{-(n-1)} &  \cdots & q^{-(n-1)(n-1)} & q^{-n(n-1)}\\
0 & 0 &  \cdots & q^{-(n-1)(n-2)} & q^{- n(n-2)}\\
\vdots & \vdots & \ddots & \vdots & \vdots\\
0 & 0 & \cdots & q^{-(n-1)} & q^{-n}\\
0 & 0 & \cdots & 0 & 1
\end{array} \right],
\end{equation*}
and 
\begin{equation*}
\tilde{\mathbf{b}}=\left[\begin{array}{cccc} 1 & q^{(T-n)} & \cdots & q^{n(T-n)}  \end{array}\right]^\mathrm{T}.
\end{equation*}
The remaining conditions in this case can be written as 
\begin{equation*}
\tilde{\mathbf{A}}\boldsymbol{\alpha}^* \succeq 2^{-C_{\textsl{s}}+o(1)} \tilde{\mathbf{b}},
\end{equation*}
which is exactly similar to \eqref{eq:KuhnTuckerCond_MatrixForm}, for $\delta=n$. 
Therefore, the optimal solution for the first case will also 
satisfy these conditions, {\em i.e.},
\begin{equation}\label{eq:Solution_Form_small_n_2}
\alpha^*_i= \left\{\begin{array}{ll} 
q^{i(T-i)} 2^{-C_{\textsl{s}}+o(1)}  & \kappa \le i\le n,\\
0 &  0\le i < \kappa,
\end{array}\right.
\end{equation}
with $\kappa=\min[(T-n)^+,n]$. Summarizing \eqref{eq:Solution_Form_small_n_1} 
and \eqref{eq:Solution_Form_small_n_2}, we can obtain the optimal solution 
for this regime, as
\begin{equation}\label{eq:Solution_Form_small_n_final}
\alpha^*_i=\left\{ \begin{array}{ll}
0  & n<i\le\delta,\\
q^{i(T-i)} 2^{-C_{\textsl{s}}+o(1)}  & \kappa\le i\le n,\\
0 &  0\le i<\kappa,
\end{array} \right.
\end{equation}
where $\kappa=\min[(T-n)^+,n]$. This completes the proof of 
Theorem~\ref{thm:Main_Result_Single_Src_Dist}.
By normalizing $\alpha^*_i$ to $1$ we can also obtain an alternative proof to 
Theorem~\ref{thm:Main_Result_Single_Src}.

\noindent{\bf Discussion:} To characterize the exact value of
$q_0$ one have to consider the exact form of the set of equations 
given in \eqref{eq:KuhnTucker_Cond_dim_l} (for each $l$)
which are as follows,
\begin{align*}
A_{ll}\alpha^*_l \ge  q^{l(T-n)} 2^{-C_{\textsl{s}}+\epsilon_q(l)} \left[1 - \sum_{j=l+1}^{\delta} q^{(T-n-j)(j-l)} 2^{[\epsilon_q(j)-\epsilon_q(l)]}\right].
\end{align*}
Although it is hard to find $q_0$ exactly, it is possible to show 
that there exists finite $q_0$ such that result of 
Theorem~\ref{thm:Main_Result_Single_Src_Dist} holds for.
This can be done by solving above equations assuming that
$\epsilon_q(k)$ is zero for every $k$ (assuming $q\gg 1$).
Then, it can be observed that the RHS of \eqref{eq:KuhnTucker_Cond_dim_l}
are either greater or less than zero. Now by assuming 
finite but large enough $q$ and considering the exact form
of \eqref{eq:KuhnTucker_Cond_dim_l} we have some small perturbations
that cannot change the sign of RHS of \eqref{eq:KuhnTucker_Cond_dim_l}
so we are done.

\subsection{Proof of Theorem~\ref{thm:Main_Result_Single_Src_Extnd}}
\label{sec:Proof_Main_Result_Single_Src_Extnd}
Let $\epsilon_q(k)$ denotes the error term in \eqref{eq:I_Derivative_Approx}.
We can easily write the exact expression for $\epsilon_q(k)$ which 
is as follows
\begin{align*}
\epsilon_q(k) =& - \sum_{d_y=0}^{r_k} \psi(n,d_y) \gaussnum{k}{d_y} q^{-nk} \log_2\left( \sum_{d_x=d_y}^{\min[m,T]} \alpha_{d_x} \frac{\gaussnum{d_x}{d_y}}{\gaussnum{T}{d_y}} q^{-nd_x} \right) \nonumber\\
&+ \log_2\left( \sum_{d_x=r_k}^{\min[m,T]} q^{r_k(d_x-r_k)-nd_x} \alpha_{d_x} \right) - r_k(T-r_k)\log_2{q},
\end{align*}
where $r_k=\min[n,k]$. 

We consider the case where $T>n+\min[m,n]$ so
Theorem~\ref{thm:Main_Result_Single_Src_Dist} implies that for the
optimal input distribution we have $\alpha_{i^*}=1$ where
$i^*=\min[m,n]$ and $q>q_0$.  Then we can simplify $\epsilon_q(k)$
more and write
\begin{align}\label{eq:ErrorTerm_epsilon_2}
\epsilon_q(k) =& \sum_{d_y=0}^{r_k} \psi(n,d_y) \gaussnum{k}{d_y} q^{-nk} \log_2\left( \frac{\gaussnum{T}{d_y}}{\gaussnum{i^*}{d_y}} \right) -r_k(T-i^*),
\end{align}
where we also use Lemma~\ref{lem:psi_Recursive} in the above simplification.

To find $q_0$, the minimum value of $q$ that the result of 
Theorem~\ref{thm:Main_Result_Single_Src_Extnd} 
is valid for, we should consider the exact form of \eqref{eq:KuhnTucker_Cond_dim_l}
and check that the RHS of \eqref{eq:KuhnTucker_Cond_dim_l} is less than
or equal to zero for $0\le l\le (i^*-1)$.
So from \eqref{eq:KuhnTucker_Cond_dim_l} for every $0\le l\le (i^*-1)$ we may write
\begin{align*}
 \left[1 -  q^{(T-n-i^*)(i^*-l)} 2^{\left[\epsilon_q(i^*)-\epsilon_q(l)\right]}\right]\le 0,
\end{align*}
or equivalently
\begin{equation}\label{eq:Set_of_CondsOn_q0_1}
\frac{\epsilon_{q_0}(l)-\epsilon_{q_0}(i^*)}{(T-n-i^*)(i^*-l)} \le \log_2{q_0}, \quad \forall l: 0\le l\le (i^*-1).
\end{equation}
Using a similar argument we should have also
\begin{equation}\label{eq:Set_of_CondsOn_q0_2}
\frac{\epsilon_{q_0}(l)-\epsilon_{q_0}(i^*)}{i^*(l-i^*)} \le \log_2{q_0}, \quad \forall l: (i^*+1)\le l\le m.
\end{equation}

From \eqref{eq:Solution_Form_small_n_final} for the capacity $C_{\textsl{s}}$
we have $C_{\textsl{s}}=i^*(T-i^*)\log_2{q}+\epsilon_q(i^*)$. Evaluating 
\eqref{eq:ErrorTerm_epsilon_2} at $k=i^*$ we have
\begin{align*}
\epsilon_q(i^*) =& \sum_{d_y=0}^{i^*} \psi(n,d_y) \gaussnum{i^*}{d_y} q^{-ni^*} \log_2\left( \frac{\gaussnum{T}{d_y}}{\gaussnum{i^*}{d_y}} \right) 
-i^*(T-i^*)\log_2{q},
\end{align*}
which results in the capacity stated in the assertion of Theorem~\ref{thm:Main_Result_Single_Src_Extnd}.

\noindent{\bf Discussion:} We derive a sufficient condition on the minimum size of $q$
to satisfy the set of conditions stated in \eqref{eq:Set_of_CondsOn_q0_1} and 
\eqref{eq:Set_of_CondsOn_q0_2}. Using this sufficient condition we
explore the behavior of $q_0$ as $T$ increases.

For $k\neq i^*$ we can write
\begin{align}\label{eq:epsilon_UpperBound}
\epsilon_q(k) &\stackrel{(a)}{\le} 4 \sum_{d_y=0}^{r_k} q^{-(n-d_y)(k-d_y)} \log_2\left(4 q^{d_y(T-i^*)} \right) -r_k(T-i^*)\log_2{q} \nonumber\\
&\le 8 +4r_k q^{-(\max[n,k]-\min[n,k]+1)} \left( 2+ (r_k-1)(T-i^*)\log_2{q} \right) \nonumber\\
&\stackrel{(b)}{\le} (8 + 8r_k) + \left( 4r_k(r_k-1)(T-i^*) \frac{\log_2{q}}{q^{(\max[n,k]-\min[n,k]+1)}} \right),
\end{align}
where $(a)$ follows from \eqref{eq:GaussianNumber_bound_1} and \eqref{eq:psi_bound},
and in $(b)$ we use the fact that $k\neq i^*$.

Then for $k=i^*$ we can write
\begin{align}\label{eq:epsilon_LowerBound}
\epsilon_q(i^*) &\ge \psi(n,i^*)q^{-ni^*} \log_2{\gaussnum{T}{i^*}} -i^*(T-i^*)\log_2{q} \nonumber\\
& \stackrel{(a)}{\ge} -(i^*)^2 (T-i^*)\frac{\log_2{q}}{q^{n-i^*+1}},
\end{align}
where $(a)$ follows from \eqref{eq:GaussianNumber_bound_1} and \eqref{eq:psi_bound}.

Let us consider two cases. First, we assume that $m\le n$ so $i^*=m$. To find a
sufficient condition for $q_0$ we have to only consider conditions given in 
\eqref{eq:Set_of_CondsOn_q0_1}. Using \eqref{eq:epsilon_UpperBound} and 
\eqref{eq:epsilon_LowerBound} and assuming that $T\rightarrow\infty$ we should have
$\log_2{q_0} \ge 5m^2 q_0^{-n+m-1}\log_2{q_0}$,
or equivalently $q_0^{n-m+1}\ge 5(i^*)^2$.

For the second case we have $m>n$ which means $i^*=n$. Here, using a similar 
argument to the one given above for the first case we can show that conditions 
\eqref{eq:Set_of_CondsOn_q0_1} give some constant $q_0$ as $T\rightarrow\infty$.
However, the conditions \eqref{eq:Set_of_CondsOn_q0_2} give a sufficient condition
for $q_0$ which grows as $T\rightarrow\infty$. Now, using \eqref{eq:Set_of_CondsOn_q0_2},
\eqref{eq:epsilon_UpperBound}, and 
\eqref{eq:epsilon_LowerBound} and assuming that $T\rightarrow\infty$, a 
sufficient condition for $q_0$ would be 
$\log_2{q_0} \ge 4nT q_0^{-2}\log_2{q_0} + nT q_0^{-1}\log_2{q_0}$. For
large $T$ for the sufficient condition we have $q_0\ge i^*T$.

\section{Multiple Sources Scenario: The Rate Region} \label{sec:ProofMultiSource}

The goal of this section is to characterize $\mathcal{R}$, the 
set of all achievable rate pairs $(R_1,R_2)$ for two user 
communication over the multiple access channel $\Cca_{\textsl{m-MAC}}$ 
described in Definition~\ref{def:matrix_channel_MAC}. More precisely, 
we will show that $\mathcal{R}=\mathcal{R}^*$. In order to do this, 
we first formulate a mathematical model for this channel. Then, we 
present an achievability scheme, to show that $\mathcal{R}^*$ is 
achievable, \emph{i.e.,} $\mathcal{R}^* \subseteq \mathcal{R}$. In 
the next subsection we prove the optimality of this scheme and show 
that $\mathcal{R} \subseteq \mathcal{R}^*$. 

The proof of the converse part of the theorem is based on two outer 
bounds, namely, a cooperative bound and a coloring bound. For the coloring 
bound, we utilize a combinatorial argument to bound the number of 
\emph{distinguishable} symbol pairs that can be transmitted from the 
two sources to the destination. This bound allows us to restrict the 
\emph{effective} input alphabets of the sources to subsets of the 
original alphabets, with significantly smaller size. We  can then 
easily bound the capacity region of the network using the restricted 
input alphabet.

The transition probability of the channel given by Definition~\ref{def:matrix_channel_MAC}, 
$P_{Y|X_1 X_2}$, can be written as \cite{JaFrDi-ISIT08}
\begin{align}
P_{Y|X_1 X_2}(y|x_1,x_2)
=\left\{ \begin{array}{ll}
q^{-n\dim(\sspan{x_1} + \sspan{x_2})} & \sspan{y} \sqsubseteq \sspan{x_1} + \sspan{x_2},\\
0 & \text{otherwise}.
\end{array} \right.
\label{eq:MAC-transition-matrix}
\end{align}

Our first result, stated in Theorem~\ref{thm:channel_equivalence_MAC}, is 
that the multiple access matrix channel described in Definition~\ref{def:matrix_channel_MAC} is
equivalent to the ``subspace'' channel 
$\mathrm{Ch}_{\textsl{s-MAC}}$ described in Definition~\ref{def:subspace_channel_MAC}, 
that has subspaces as inputs and outputs. So to characterize the optimal rate
region of $\mathrm{Ch}_{\textsl{m-MAC}}$, we can focus on finding the optimal
rate region of $\mathrm{Ch}_{\textsl{s-MAC}}$. We will use this equivalence  in the rest of this section.

We know from \cite{CoTh-ElmntsInfoTheory06} that the rate region of the multiple 
access channel $\mathrm{Ch}_{\textsl{s-MAC}}$ is given by the closure of the 
convex hull of the rate vectors satisfying
\begin{equation*}
R_S\le I(\Pi_{X_S};\Pi_Y|\Pi_{X_{S^c}}) \quad \text{for all } S\subseteq\{1,\ldots,N_s\},
\end{equation*}
for some product distribution $P_{\Pi_{X_1}}(\pi_1)\cdots P_{\Pi_{X_{N_s}}}(\pi_{N_s})$. 
Note that $R_S=\sum_{i\in S} R_i$, where $R_i$ is the transmission 
rate of the $i$th source, $\Pi_{X_S}=\{\Pi_{X_i}:i\in S\}$ and 
$S^c$ is the complement set of $S$.

\subsection{Achievability Scheme}\label{sec:Multiple_Src_Achv}

In this subsection we illustrate a simple achievability 
scheme for the corner points of the rate region defined in 
Theorem~\ref{thm:Main_Result_Multpl_Src}. The remaining
points in the rate region can be achieved using time-sharing. 

For given $(d_1,d_2)\in\mathcal{D}^*$, define the following 
subspace code-books
\begin{equation*}
\begin{aligned}
&\wt{\Cca}_1 \triangleq \bigg\{\sspan{X_1}: X_1=
 \left[ \begin{array}{c|c|c}
\mathbf{I}_{d_1\times d_1} & \mathbf{0}_{d_1 \times d_2}  & \mathbf{U}_1\\
\hline
\mathbf{0}_{(m_1-d_1)\times d_1} & \mathbf{0}_{(m_1-d_1) \times d_2}  & \mathbf{0}_{ (m_1-d_1) \times (T-d_1-d_2) }
\end{array} \right],
\mathbf{U}_1\in\Fbb_q^{d_1\times (T-d_1-d_2)}
\bigg\}
\end{aligned}
\end{equation*}
and 
\begin{align*}
&\wt{\Cca}_2 \triangleq \bigg\{\sspan{X_2}: X_2=
\left[ \begin{array}{c|c|c}
\mathbf{0}_{d_2\times d_1} & \mathbf{I}_{d_2 \times d_2}  & \mathbf{U}_2\\
\hline
\mathbf{0}_{(m_2-d_2)\times d_1} & \mathbf{0}_{(m_2-d_2) \times d_2}  & \mathbf{0}_{ (m_2-d_2) \times (T-d_1-d_2) }
\end{array} \right],
\mathbf{U}_2\in\Fbb_q^{d_2\times (T-d_1-d_2)}
\bigg\}.
\end{align*}
If we transmit messages from these code-books, we have
\begin{align*}
Y&=H_1X_1+H_2X_2\\
&=\left[ \begin{array}{c|c|c}
\hat{H}_1 & \hat{H}_2 & \hat{H}_1 \mathbf{U}_1 + \hat{H}_2 \mathbf{U}_2
\end{array} \right],
\end{align*}
where $\hat{H}_i$ captures the first $d_i$ columns of $H_i$. 
Therefore, decoding at the receiver would be just recovering
of $\mathbf{U}_1$ and $\mathbf{U}_2$ given
$\hat{H}_1 \mathbf{U}_1 + \hat{H}_2 \mathbf{U}_2$, $\hat{H}_1$, 
and $\hat{H}_2$. Since $d_1+d_2\leq n$, the matrix 
$[\hat{H}_1 \ \hat{H}_2]$ is full-rank with high probability, and 
therefore the decoder is able to decode $\mathbf{U}_1$ and $\mathbf{U}_2$. 

Note that the achievability scheme uses effectively the coding 
vectors approach \cite{ChWuJa-Alert03}. This indicates that 
for $\frac{T}{2}>\max [m_1+m_2,n]$ and $q$ large enough, the 
subspace coding and the coding vectors approach achieve the 
same rate.

\subsection{Outer bound on the Admissible Rate Region}\label{sec:outer}

In the following we will present an outer bound for $\mathcal{R}$, the 
admissible rate region of the non-coherent two-user multiple access 
channel $\mathrm{Ch}_{\textsl{m-MAC}}$. Recall that by 
Theorem~\ref{thm:channel_equivalence_MAC} we can focus on the subspace 
channel $\mathrm{Ch}_{\textsl{s-MAC}}$. We first show in 
Proposition~\ref{pre:OuterBoundCoop} that $\Rca\subseteq \Rca_{\mathrm{coop}}$, a cooperative outer-bound. Then Proposition~\ref{pre:OuterBoundColering} 
demonstrates that $\Rca\subseteq\Rca_{\mathrm{col}}$, a coloring outer-bound. 
Finally we show that $\Rca_{\mathrm{col}}\cap\Rca_{\mathrm{coop}}\subseteq\Rca$, 
yielding the desired outer-bound $\Rca\subseteq\Rca^*$ which matches the 
achievability of \S\ref{sec:Multiple_Src_Achv}.

The first outer bound, called cooperating outer bound, is 
simply obtained by letting the two transmitters cooperate to 
transmit their messages to the receiver, i.e. we assume they form a super-source. Applying 
Theorem~\ref{thm:Main_Result_Single_Src} for the non-coherent 
scenario for the single super-source, the one who controls 
the packets of both transmitters, we have the following 
proposition.
\begin{proposition}\label{pre:OuterBoundCoop}
Let $\frac{T}{2} \geq m_1+m_2$. We have $\mathcal{R} \subseteq \mathcal{R}_{\mathrm{coop}}$ where 
\begin{equation*}
\mathcal{R}_{\mathrm{coop}}\triangleq\left\{ (R_1,R_2):\ R_1+R_2\le k(T-k) \log_2 q \right\},
\end{equation*}
and $k=\min[m_1+m_2,n]$.
\end{proposition}

The rest of this section is dedicated to deriving the second outer 
bound which is denoted by $\mathcal{R}_{\mathrm{col}}$. This 
bound is based on an argument on the number of messages per 
channel use that each user can reliably communicate over 
the multiple access channel.

Let $(R_1,R_2)\in\mathcal{R}$  be an achievable rate pair for 
which there exists an encoding and decoding scheme with block 
length $N$ and small error probability.  One can follow the 
usual converse proof of the multiple access channel 
from \cite{CoTh-ElmntsInfoTheory06} to show that 
\begin{align}
R_1 &\leq  I(\Pi_{X_1}^N; \Pi_{Y}^N | \Pi_{X_2}^N) \leq \frac{1}{N} \sum_{t=1}^{N} I(\Pi_{X_1t}; \Pi_{Yt} | \Pi_{X_2t}),\nonumber\\
R_2 &\leq  I(\Pi_{X_2}^N; \Pi_{Y}^N | \Pi_{X_1}^N) \leq \frac{1}{N} \sum_{t=1}^{N} I(\Pi_{X_2t}; \Pi_{Yt} | \Pi_{X_1t}),\nonumber\\
R_1+R_2 &\leq  I(\Pi_{X_1}^N, \Pi_{X_2}^N ; \Pi_{Y}^N ) \leq \frac{1}{N} \sum_{t=1}^{N} I(\Pi_{X_1t}, \Pi_{X_2t} ; \Pi_{Yt}).\nonumber
\end{align}
For each time instance $t$, denote by $\wt\Cca_{i,t}$, the 
projection of the code-book used by user $i$ to its $t$-th 
element. For a single source scenario, we have shown 
in \S\ref{sec:TheChanlCap-SnglSrc} that we can use the 
set $\mathrm{Sp}(T,m)$ as our input alphabet for all time 
slots, and have the receiver successfully decode the sent 
messages, and hence, the user can communicate $\mathcal{S}(T,m)$ distinct 
messages. For the multi-source case, $\wt\Cca_{i,t}$ is  more 
restricted. The main reason for this is that the transition 
probability of the multiple access channel 
$P_{\Pi_Y|\Pi_{X_1}\Pi_{X_2}}$ is of the form $P_{\Pi_Y|\Pi_{X_1} + \Pi_{X_2}}$. 
That is, if $(\pi_1,\pi_2)\in\wt\Xca_1\times\wt\Xca_2$ and 
$(\pi_1',\pi'_2)\in\wt\Xca_1\times\wt\Xca_2$ satisfy 
$\pi_1+ \pi_2 =\pi'_1+ \pi'_2$, then $P(\Pi_Y|\pi_1,\pi_2)=P(\Pi_Y|\pi_1', \pi_2')$, 
and hence the receiver cannot distinguish between  the two pairs. 

In the following we will discuss this indistinguishability 
in detail, and derive the maximum number of distinguishable 
pairs which can be conveyed through the channel. In order to 
do so,  we start with some useful definitions and lemmas.  

\begin{definition}
For a fixed $\pi_1\in \mathrm{Gr}(T,d_1)$, we denote by 
$\mathcal{N}(\pi_1,d_2,d_{12})$ the set of subspaces of 
dimension $d_2$ that intersect with $\pi_1$ at $d_{12}$ 
dimensions, \emph{i.e.},
\begin{align}
\mathcal{N}(\pi_1,d_2,d_{12})\triangleq \{\pi_2\in \mathrm{Gr}(T,d_2): \dim(\pi_1 \cap \pi_2)=d_{12}\}.
\end{align}
\end{definition}
It turns out that the cardinality of the set  $\mathcal{N}(\pi_1,d_2,d_{12})$ 
depends on $\pi_1$ only through its dimension, $d_1=\dim(\pi_1)$. 
Therefore, we denote this number by $n(d_1,d_2,d_{12})$, which is 
characterized in the following lemma.
\begin{lemma}\label{lem:N}
The cardinality of the set $N(\pi_1,d_2,d_{12})$ is given by 
\begin{align}
n(d_1,d_2,d_{12})=|N(\pi_1,d_2,d_{12})|\=q^{d_{12}(d_1-d_{12}) + (d_2-d_{12})(T-d_2)}.
\end{align}
\end{lemma}

\begin{definition}
For a fixed $\pi_1\in \mathrm{Gr}(T,d_1)$ and $\pi_2\in \mathrm{Gr}(T,d_2)$, we define
\begin{align}
A(\pi_1,\pi_2)\triangleq \{\pi'_2 \in \mathrm{Gr}(T,d_2):  \pi_1 + \pi'_2=\pi_1 + \pi_2\}.
\end{align}
\end{definition}

\begin{lemma}\label{lem:A}
The cardinality of the set $A(\pi_1,\pi_2)$ only depends on 
the dimensions of the two subspaces and their intersection, 
$d_1=\dim(\pi_1)$, $d_2=\dim(\pi_2)$, and $d_{12}=\dim(\pi_1 \cap \pi_2)$. 
Moreover, it can be asymptotically characterized by
\begin{align}
a(d_1,d_2,d_{12})=|A(\pi_1,\pi_2)|\=q^{ d_2(d_1-d_{12}) }.
\end{align}.
\end{lemma}

\begin{definition}\label{def6}
For an arbitrary set $\wt\Cca\subseteq \mathrm{Sp}(T,m)$, we 
denote the projection of $\wt\Cca$ onto the set of $d$-dimensional 
Grassmannian $\wt\Cca(d)$. Formally, 
\begin{align*}
\wt\Cca(d)\triangleq \wt\Cca \cap \mathrm{Gr}(T,d)=\{ \pi\in \wt\Cca: \dim(\pi)= d\}.
\end{align*}
\end{definition}

For a fixed time instance $t$, and corresponding subsets 
$\wt\Cca_{1,t}$ and $\wt\Cca_{2,t}$, we can construct a 
table with $|\wt\Cca_{1,t}|$ rows and $|\wt\Cca_{2,t}|$ 
columns, each row (column) corresponding to one subspace 
$\pi_1$ ($\pi_2$) in $\wt\Cca_{1,t}$ ($\wt\Cca_{2,t}$). In 
the following, we define an equivalence relation for the 
cells of this table. 

\begin{definition}\label{def:col}
A \emph{coloring} for a table constructed as above is an 
assignment of colors to the cells of the table using a 
function $\col:\wt\Cca_{1,t}\times \wt\Cca_{2,t} \rightarrow \mathbb{N}$ 
such that $\col(\pi_1,\pi_2)=\col(\pi_1',\pi_2')$ if 
and only if $\pi_1+ \pi_2=\pi_1' + \pi_2'$.
\end{definition}
It is clear that the coloring definition above exactly 
matches with that of indistinguishability we discussed 
before. More precisely, two pairs of subspaces $(\pi_1,\pi_2)$ 
and $(\pi_1',\pi_2')$ are distinguishable if and only 
if their corresponding cells in the table have different 
colors. The following theorem upper bounds the cardinality 
of the subspace sets based on this fact.

\begin{theorem}\label{thm:OuterBoundColering}
For each pair of uniquely distinguishable sets 
$(\wt\Cca_{1,t},\wt\Cca_{2,t})$ defined on the input 
alphabet $\wt\Xca_1\times\wt\Xca_2$ for the multiple access 
channel $\mathrm{Ch}_{\textsl{s-MAC}}$, there exist 
integer numbers $0\leq \delta_i(t) \leq m_i$ such that 
\begin{align}
|\wt\Cca_{i,t}|\aleq q^{\delta_i(t) \left(T-\delta_1(t)-\delta_2(t)\right)}, \qquad i=1,2.
\end{align}
\end{theorem}

\begin{proof}
We may drop the time index $t$ in this proof for brevity. 
For a fixed $t$, let $\delta_i$ be the \emph{dominating} 
dimension in the set $\wt\Cca_i$, {\em i.e.}, 
\begin{align*}
\delta_i\triangleq \arg \max_{d} |\wt\Cca_i(d)|, 
\end{align*}
where $\wt\Cca_i(d)$ is as defined in Definition~\ref{def6}. 
It is clear that 
\begin{align}
|\wt\Cca_i|=\sum_d |\wt\Cca_i(d)| \leq m_i |\wt\Cca_i(\delta_i)|\=|\wt\Cca_i(\delta_i)|,
\label{eq:al-al}
\end{align}
where the last asymptotic equality follows from the fact that $m_i$ 
is a constant with respect to the underlying field size $q$.
This means that we may lose only a constant factor in the 
code-book size by removing all subspaces from 
$\wt\Cca_1$ ($\wt\Cca_2$), except the ones that have 
dimension $\delta_1$ ($\delta_2$) . Therefore the loss in 
the rate values would be  negligible as $q$ grows. Consider 
the table constructed for $\wt\Cca_1(\delta_1)$ and 
$\wt\Cca_2(\delta_2)$. Let $\pi_1\in\wt\Cca_1(\delta_1)$ be 
a $\delta_1$-dimensional subspace, and consider the 
corresponding row of the table. We further partition the 
columns of the table with respect to $\pi_1$ into  
$\bigcup_{d_{12}=0}^{\min[\delta_1,\delta_2]}\wt\Cca_2(\pi_1,\delta_2,d_{12})$, where 
\begin{align}
\wt\Cca_2(\pi_1,\delta_2,d_{12})\triangleq\{\pi_2\in\wt\Cca_2(\delta_2): \dim(\pi_1\cap \pi_2)=d_{12} \}.
\end{align}
We use $K(\pi_1,\delta_2)$ and $K(\pi_1,\delta_2,d_{12})$ to 
denote the number of different colors in the row that corresponds 
to $\pi_1$ and its intersection with $\wt\Cca_2(\pi_1,\delta_2,d_{12})$, respectively. 

Note that $\wt\Cca_2(\pi_1,\delta_2,d_{12})\subseteq \mathcal{N}(\pi_1,\delta_2,d_{12})$, 
and therefore the number of different colors that appear in 
this partition of the row, cannot exceed the number of 
colors that could potentially appear if $\mathcal{N}(\pi_1,\delta_2,d_{12})\subseteq \wt\Cca_2$. 
Recall that $\mathcal{N}(\pi_1,\delta_2,d_{12})$ has 
$n(\delta_1,\delta_2,d_{12})$ elements, which are split 
into subsets of size $a(\delta_1,\delta_2,d_{12})$ of 
the same color. Therefore, for a large field size, the number 
of different colors in this partition of the row 
corresponding to $\pi_1$, can be upper bounded as 
\begin{align}
K(\pi_1,\delta_2,d_{12}) \leq \frac{n(\delta_1,\delta_2,d_{12})}{a(\delta_1,\delta_2,d_{12})} \= q^{(\delta_2-d_{12})(T-\delta_1-\delta_2+d_{12})}. 
\end{align}
Hence, 
\begin{align}
K(\pi_1,\delta_2) &= \sum_{d_{12}=0}^{\min[\delta_1,\delta_2]} K(\pi_1,\delta_2, d_{12})\nonumber\\
&\aleq \sum_{d_{12}=0}^{\min[\delta_1,\delta_2] } q^{(\delta_2-d_{12})(T-\delta_1-\delta_2+d_{12})}\nonumber\\
&\= q^{\max_{0\leq d_{12} \leq \min[\delta_1,\delta_2]} (\delta_2-d_{12})(T-\delta_1-\delta_2+d_{12})}\nonumber\\
&= q^{\delta_2(T-\delta_1-\delta_2)}
\end{align}
where the asymptotic inequality and equality hold for 
large $q$. Moreover, the last equality is based on the 
assumption $T\geq 2(m_1+m_2) \geq 2(\delta_1+\delta_2)$ and  
the fact that the exponent is a decreasing function of $d_{12}$ 
for $0\leq d_{12} \leq \min[\delta_1,\delta_2]$. 

It is worth mentioning that this argument holds for each 
choice of $\pi_1 \in \wt\Cca_1(\delta_1)$. This means if 
the first user transmits a $\delta_1$-dimensional subspace, the 
receiver cannot distinguish more that $q^{\delta_2(T-\delta_1-\delta_2)}$ 
different symbols. The same argument holds for a fixed column 
$\pi_2\in\wt\Cca_2$ which yields an upper bound to the number 
of distinguishable messages as $q^{\delta_1(T-\delta_1-\delta_2)}$.  
\end{proof}

Theorem~\ref{thm:OuterBoundColering} essentially upper bounds the single 
letter mutual information $I(\Pi_{X_1t} ; \Pi_{Yt} | \Pi_{X_2 t})$ for 
any time instance $t$. The following proposition summarizes this discussion.

\begin{proposition}\label{pre:OuterBoundColering}
We have $\mathcal{R}\subseteq\mathcal{R}_{\mathrm{col}}$ where
\begin{equation*}
\mathcal{R}_{\mathrm{col}}\triangleq \mathrm{convex~hull} \bigcup_{(d_1,d_2)\in\mathcal{D}_{\mathrm{col}}} \mathcal{R}(d_1,d_2),
\end{equation*}
in which $\mathcal{R}(d_1,d_2)$ is as defined in \eqref{eq:def:Ri}, and 
\begin{align*}
\mathcal{D}_{\mathrm{col}}\triangleq\{(d_1,d_2):\ 0\leq d_i \leq m_i\}.
\end{align*}
\end{proposition}
\begin{proof}
Using Theorem~\ref{thm:OuterBoundColering}, we can upper bound the 
number of distinguishable pairs for each time instance. For a fixed~$t$, let $\delta_1(t)$ and $\delta_2(t)$ denote the dominating dimensions. 
Therefore, we have
\begin{align}
R_1 &\leq \frac{1}{N} \sum_{t=1}^{N} I(\Pi_{X_1t}; \Pi_{Yt} | \Pi_{X_2t}),\nonumber\\
&\stackrel{\cdot}{\leq} \frac{1}{N} \sum_{t=1}^{N} \log_2 q^{[\delta_1(t)(T-\delta_1(t)-\delta_2(t))]}\nonumber\\
&= \frac{1}{N} \sum_{t=1}^{N} \delta_1(t)(T-\delta_1(t)-\delta_2(t)) \log_2 q,\nonumber
\end{align}
where $0\leq \delta_i(t) \leq m_i$ for $t=1,\dots,N,$ and $i=1,2$. Similarly, we have
\begin{align}
R_2 \leq \frac{1}{N} \sum_{t=1}^{N} \delta_2(t)(T-\delta_1(t)-\delta_2(t)) \log_2 q.\nonumber
\end{align}
Therefore, 
\begin{align}
(R_1,R_2) \leq \frac{1}{N} \sum_{t=1}^{N} \left( \delta_1(t)(T-\delta_1(t)-\delta_2(t)) \log_2 q, \delta_2(t)(T-\delta_1(t)-\delta_2(t)) \log_2 q\right).
\label{eq:R12-col}
\end{align}
It is clear that the RHS of \eqref{eq:R12-col} is a convex linear combination of the points 
\begin{align*}
\left\{\delta_1(t)(T-\delta_1(t)-\delta_2(t))\log_2 q , \delta_1(t)(T-\delta_1(t)-\delta_2(t))\log_2 q \right\}_{t=1}^{N} 
\end{align*}
which are in the region $\mathcal{R}(\delta_1(t),\delta_2(t))$. This completes the proof. 
\end{proof}

Summarizing Proposition~\ref{pre:OuterBoundCoop} and 
Proposition~\ref{pre:OuterBoundColering}, we have 
$\mathcal{R} \subseteq \Rca_{\mathrm{coop}} \cap \Rca_{\mathrm{col}}$.  
So, it only remains to prove the following theorem in order 
to show that $\mathcal{R}^*$ is an outer bound for the admissible 
rate region.
 
\begin{theorem}\label{thm:MAC-OB}
We have $\Rca_{\mathrm{coop}}\cap \Rca_{\mathrm{col}}\subseteq \Rca^*$.
\end{theorem}
Before presenting the proof of the theorem, we give the 
following two lemmas, which help us to characterize the corner 
points of the region of our interest. 

\begin{lemma}\label{lem:corner_col}
The set of corner points
of $\Rca_{\mathrm{col}}$ is the set of all rate pairs of the form 
\[
(R_1,R_2)= \left(R_1(d_1,d_2),R_2(d_1,d_2)\right),
\] 
for some $(d_1,d_2)\in\widetilde{\mathcal{D}}$, where
\begin{align*}
\widetilde{\mathcal{D}}=\{&(0,m_2),(1,m_2),\dots,(m_1,m_2),
(m_1,m_2-1),\dots,(m_1,1),(m_1,0)\}.
\end{align*}
\end{lemma}

\begin{lemma}\label{lem:intersect_col_coop}
If $\Rca_{\mathrm{col}} \nsubseteq \Rca_{\mathrm{coop}}$, then 
any intersecting point of $R_1+R_2=k(T-k)\log_2 q$ with the 
boundary of $\Rca_{\mathrm{col}}$ is a point of the form  
$(R_1(d_1,d_2), R_2(d_1,d_2))$, where 
\begin{align*}
(d_1,d_2) \in \widetilde{\mathcal{D}} \cup \{(m_1-1,0),\dots,(0,0),(0,1),\dots,(0,m_2-1)\}.
\end{align*}
That is, the boundaries of $\Rca_{\mathrm{col}}$ 
and $\Rca_{\mathrm{coop}}$ can only intersect on either the 
corner points of $\Rca_{\mathrm{col}}$ or the $R_1-R_2$ axes. 
\end{lemma}

\begin{proof}[Proof of Theorem~\ref{thm:MAC-OB}]
Note that $ \Rca_{\mathrm{coop}} \cap \Rca_{\mathrm{col}}$ 
is a convex polytope, formed as intersection of a polytope 
and the convex hull of a finite number of polytopes. 
Therefore, it suffices to prove the theorem only for its 
corner points. Let $(R_1,R_2)\in \Rca_{\mathrm{coop}} \cap \Rca_{\mathrm{col}}$ 
be a corner point. It is clear that one of the followings occurs. 
\begin{description}
\item{(i)} $(R_1,R_2)$ is a corner point of $\Rca_{\mathrm{col}}$ 
and interior point of $\Rca_{\mathrm{coop}}$;
\item{(ii)}  $(R_1,R_2)$ is an intersecting point of the 
boundaries of $\Rca_{\mathrm{col}}$ and $\Rca_{\mathrm{coop}}$. 
\end{description}
In the former case, Lemma~\ref{lem:corner_col} which characterizes 
the set of corner points of $\Rca_{\mathrm{col}}$, implies there 
exists a pair $(d_1,d_2)\in\wt{\mathcal{D}}$ such that 
$(R_1,R_2)=(R_1(d_1,d_2),R_2(d_1,d_2))$.  Also 
$(R_1,R_2)\in\Rca_{\mathrm{coop}}$ implies
\[(d_1+d_2)(T-(d_1+d_2))\log_2 q = R_1+R_2 \leq k(T-k) \log_2 q.\]
Note that the function $f(x)\triangleq x(T-x)$ is an 
increasing function of $x$ for $x\in(0,T/2)$. Therefore, 
$d_1+d_2\leq k = \min\{m_1+m_2, n\}$, and hence 
$(d_1,d_2)\in\mathcal{D}^*$, which implies that $(R_1,R_2)\in\Rca^*$. 

In the latter case, it follows from Lemma~\ref{lem:intersect_col_coop} 
that $(R_1,R_2)$ should be either a corner point of $\Rca_{\mathrm{col}}$ 
for which the above argument holds, or of the form 
$(R_1,R_2)=(R_1(d_1,d_2),R_2(d_1,d_2))$ with $d_1d_2=0$. 
Again $(R_1,R_2)\in\Rca_{\mathrm{coop}}$, which implies that 
$d_1+d_2\leq k= \min\{m_1,m_2,n\}$, and $(R_1,R_2)\in\Rca^*$. 
This completes the proof.
\end{proof}

\begin{corollary}\label{cor:NumCornerPointsMAC}
The number of corner points of the rate region $\Rca^*$ 
excluding the point $(0,0)$ is equal to
\[
\min\left[m_1, (n-m_2)^+\right] + \min\left[m_2, (n-m_1)^+\right] + 2 - \1_{\{n\ge m_1+m_2\}}.
\]
\end{corollary}

\begin{proof}
By Lemma~\ref{lem:corner_col} the set of corner points of 
region $\Rca_{\mathrm{col}}$ correspond to the pairs $(d_1,d_2)$ 
which belong to the set $\{(0,m_2)...(m_1,m_2)...(m_1,0)\}$. 
In this case the number of corner points excluding 
$(R_1,R_2)=(0,0)$ is $m_1+m_2+1$.

However the final rate region is the intersection of $\Rca_{\mathrm{col}}$ 
and $\Rca_{\mathrm{coop}}$, where the later one includes 
all the rate pairs with sum smaller 
than $k(T-k)\log_2{q}$, $k=\min[m_1+m_2,n]$, see 
Proposition~\ref{pre:OuterBoundCoop}.

Lemma~\ref{lem:intersect_col_coop} explains how these 
two regions intersect with each other. In this 
case, the corner points correspond to the pairs $(d_1,d_2)$ 
which belong to the set $\{(0,m_2),\ldots,(\alpha,m_2),(m_1,\beta),\ldots,(m_1,0) \}$ 
where $\alpha = \min[m_1, (n-m_2)^+]$ and $\beta = \min[m_2, (n-m_1)^+]$.
So the number of corner points excluding $(0,0)$ is
\[
\alpha + \beta + 2 - \1_{\{n\ge m_1+m_2\}},
\]
where $\1_{\{n\ge m_1+m_2\}}$ takes into account the case where 
two points $(\alpha,m_2)$ and $(m_1,\beta)$ overlap with each other.
\end{proof}

\section{Conclusions}\label{sec:Conclusion}

In this paper, we used a random matrix channel to model the problem of
multicasting over a packet network that employs randomized network
coding. We calculated the capacity of this channel for the case where
the finite field of operation $\mathbb{F}_q$ is large, but showed
through simulation results fast convergence for small values of $q$.
We prove that use of subspace coding, proposed for algebraic coding in
\cite{KoetKsch-IT08-erasure,SKK08}, is optimal for this channel. Moreover, we showed that the
capacity achieving distribution for very small packet lengths uses
subspaces of all dimensions, while as the packet length increases, the
number of required dimensions in the optimal distribution decreases.
In particular, the choice of the subspace dimension used in the
seminal work of Koetter and Kschischang \cite{KoetKsch-IT08-erasure}
is indeed optimal for large enough packet size.  We extended our work
to the case of multiple access with two sources, where we used a
coloring argument to derive an outer bound for the capacity that we
believe is interesting in itself. We showed that in all the cases we examined,
the throughput benefits subspace coding offers  as compared to the use of
coding vectors go to zero as the alphabet size $q$ increases,
and thus use of coding vectors is (asymptotically) optimal.

\section*{Acknowledgements}
The work of S.~Mohajer and C.~Fragouli was supported in part by the
ERC Starting Investigator grant \# 240317. The work of
M.~Jafari~Siavoshani and C. Fragouli was supported in part by the
Swiss National Science Foundation through the grant \#
PP002-110483. We would like to thank the anonymous reviewers for
detailed comments that greatly enhanced the paper. In particular, one
of the reviewers suggested an alternate proof for Theorem
\ref{thm:Main_Result_Single_Src}, which we have included in the paper
in \S\ref{subsec:Capacity-UpperLowerBound}. Our original proof of the
result is used in the proof of Theorem
\ref{thm:Main_Result_Single_Src_Dist} which gives a non-asymptotic
characterization.

\appendices

\section{Proofs}\label{sec:apndx1}

\begin{proof}[Proof of Theorem~\ref{thm:channel_equivalence_P2P}]
To prove the theorem, we start with $I(X;Y)$ for the channel $\mathrm{Ch}_{\textsl{m}}$, 
stated in \eqref{eq:I_1}, where the
channel transition probability is given in \eqref{eq:P2P_channel_transfer_prob_1}. 
We will show that for 
each input distribution $P_X(x)$ there exists an input distribution 
$P_{\Pi_X}(\pi_x)$ for the channel $\mathrm{Ch}_{\textsl{s}}$ such that 
$I(X;Y)=I(\Pi_Y;\Pi_X)$ and vice versa.

We know that $P_{Y|X}(y|x)=P_{Y|X}(y|x')$ if $\sspan{x}=\sspan{x'}$. 
So we can write
\begin{align*}
I(X;Y) = \sum_{\pi_x\in\wt{\Xca},\ y\in\Yca} P_{\Pi_X}(\pi_x) P_{Y|\Pi_X}(y|\pi_x)\log_2 \left( \frac{P_{Y|\Pi_X}(y|\pi_x)}{P_Y(y)} \right),
\end{align*}
where we choose $P_{\Pi_X}(\pi_x) = \sum_{x\in\Xca:\sspan{x}=\pi_x} P_X(x)$ 
and define
\[
P_{Y|\Pi_X}(y|\pi_x) \triangleq \left\{ \begin{array}{ll} q^{-n\dim(\pi_x)} & \sspan{y}\sqsubseteq \pi_x,\\ 
0 & \text{otherwise}.
\end{array} \right.
\]
Then expanding $I(X;Y)$ we have
\[
I(X;Y) = \sum_{\pi_x\in\wt{\Xca}} P_{\Pi_X}(\pi_x) \sum_{\pi_y\in\wt{\Yca}} \sum_{\begin{subarray}{c} y\in\Yca,\\ \sspan{y}=\pi_y \end{subarray}}  P_{Y|\Pi_X}(y|\pi_x)\log_2 \left( \frac{P_{Y|\Pi_X}(y|\pi_x)}{P_Y(y)} \right).
\]
Now using the symmetry properties of $P_{Y|\Pi_X}(y|\pi_x)$ we can 
simplify $I(X;Y)$. In fact $P_{Y|\Pi_X}(y_1|\pi_x)=P_{Y|\Pi_X}(y_2|\pi_x)$ 
and $P_Y(y_1)=P_Y(y_2)$ if $\sspan{y_1}=\sspan{y_2}$. So we can 
remove the summation over $y$ and write
\[
I(X;Y) = \sum_{\pi_x\in\wt{\Xca}} P_{\Pi_X}(\pi_x) \sum_{\pi_y\in\wt{\Yca}} \psi(T,n,\pi_y) P_{Y|\Pi_X}(y|\pi_x)\log_2 \left( \frac{P_{Y|\Pi_X}(y|\pi_x)}{P_Y(y)} \right),
\]
for some matrix $y$ such that $\sspan{y}=\pi_y$. Remember that 
$\psi(T,n,\pi_y)$ is defined in Definition~\ref{def:psi}, \S\ref{sec:ChannelModel-Notations}. 
Defining $P_{\Pi_Y|\Pi_X}(\pi_y|\pi_x)\triangleq \psi(T,n,\pi_y) \left. P_{Y|\Pi_X}(y|\pi_x)\right|_{\textrm{for some } y:\sspan{y}=\pi_y}$, 
we can write
\[
I(X;Y) = \sum_{\pi_x\in\wt{\Xca},\pi_y\in\wt{\Yca}} P_{\Pi_X}(\pi_x) P_{\Pi_Y|\Pi_X}(\pi_y|\pi_x) \log_2 \frac{P_{\Pi_Y|\Pi_X}(\pi_y|\pi_x)}{P_{\Pi_Y}(\pi_y)} = I(\Pi_X;\Pi_Y).
\]
Based on the above discussion going back from the channel 
$\mathrm{Ch}_{\textsl{s}}$ to $\mathrm{Ch}_{\textsl{m}}$ is very 
easy. It is sufficient to choose 
\[
P_X(x) = \frac{P_{\Pi_X}(\pi_x)}{\psi(T,m,\pi_x)},\quad \forall x:\ \sspan{x}=\pi_x,
\]
for all $\pi_x\in\wt\Xca$. This completes the proof.
\end{proof}

\begin{proof}[Proof of Lemma~\ref{lem:psi_value}]
We want to count the number of different matrices $\mathbf{X}\in\mathbb{F}_q^{n\times T}$ such that $\sspan{\mathbf{X}}=\pi_d$ where $\pi_d$ is an specific $d$ dimensional subspace of $\mathbb{F}_q^T$. 

We know that we can decompose $\mathbf{X}$ as
\begin{equation*}
\mathbf{X}=\mathbf{A}\mathbf{B},\quad \mathbf{A}\in\mathbb{F}_q^{n\times d}, \mathbf{B}\in\mathbb{F}_q^{d\times T},
\end{equation*}
where $\mathbf{A}$ and $\mathbf{B}$ are full rank matrices. Let us fix $\mathbf{B}$ such that $\sspan{\mathbf{B}}=\pi_d$. Now for every two different full rank matrices $\mathbf{A}$ and $\mathbf{A}'$ we would obtain different matrices $\mathbf{X}=\mathbf{A}\mathbf{B}$ and $\mathbf{X}'=\mathbf{A}'\mathbf{B}$ such that $\mathbf{X}\neq \mathbf{X}'$ and $\sspan{\mathbf{X}}=\sspan{\mathbf{X}'}=\pi_d$. So the number of different $\mathbf{X}$ where $\sspan{\mathbf{X}}=\pi_d$ is equal to the number of full rank $n\times d$ matrices over $\mathbb{F}$ which is equal to
$\prod_{i=0}^{d-1} (q^n-q^i)$,
and we are done.
\end{proof}


\begin{proof}[Proof of Lemma~\ref{lem:uniform_dist}]
Let $P_{\Pi_X}(\pi_x)$ be the optimal input distribution of the channel $\mathrm{Ch}_{\textsl{s}}$ with transition probabilities given in (\ref{eq:P2P_channel_transfer_prob_2}). For a fixed dimension $0\leq d\leq\min[m,T]$, and an arbitrary permutation 
\begin{align*}
\sigma: \left\{1,2,.\dots,\gaussnum{T}{d}\right\}\rightarrow \left\{1,2,.\dots,\gaussnum{T}{d}\right\}
\end{align*}
which acts on subspaces of dimension $d$, define $P_\sigma(\pi_x)$ as
\begin{align*}
P_\sigma(\pi_x)=\left\{
\begin{array}{ll}
P_{\Pi_X}(\sigma(\pi_x)) & \textrm{if $\dim(\pi_x)= d$},\\
P_{\Pi_X}(\pi_x) & \textrm{if $\dim(\pi_x)\neq d$}.
\end{array}
\right.
\end{align*}
Also define $P^*(\pi_x)=\frac{1}{\gaussnum{T}{d} !} \sum_\sigma P_\sigma(\pi_x)$ where the summation is over all possible permutations. Rewriting the mutual information in \eqref{eq:mutual_information_s_1} as a function of the input distribution and the transition probabilities, $I(P_{\Pi_X}(\pi_x) , P_{\Pi_Y|\Pi_X} (\pi_y|\pi_x))$, we have
\begin{align*}
\quad I(P^*(\pi_x), &P_{\Pi_Y|\Pi_X}(\pi_y|\pi_x)) \\
&= I\left(\frac{1}{\gaussnum{T}{d} !} \sum_\sigma P_\sigma(\pi_x), P_{\Pi_Y|\Pi_X}(\pi_y|\pi_x) \right)\\
&\stackrel{(a)}{\geq}  \frac{1}{\gaussnum{T}{d} !} \sum_\sigma I( P_\sigma(\pi_x), P_{\Pi_Y|\Pi_X}(\pi_y|\pi_x))\\
&\stackrel{(b)}{=} I(P_{\Pi_X}(\pi_x), P_{\Pi_Y|\Pi_X}(\pi_y|\pi_x))
\end{align*}
where $(a)$ is due to concavity of the mutual information with respect to the input distribution, and $(b)$ holds because $I( P_\sigma(\pi_x), P_{\Pi_Y|\Pi_X}(\pi_y|\pi_x))=I(P_{\Pi_X}(\pi_x), P_{\Pi_Y|\Pi_X}(\pi_y|\pi_x))$ for all $\sigma$, since the permutation only permutes the terms in a summation in (\ref{eq:mutual_information_s_1}).

Note that $P^*(\pi_x)$ assigns equal probabilities to all subspaces with dimension $d$, and the above-mentioned inequality shows that it is as good as the optimal input distribution. A similar argument holds for all $0\leq d \leq \min[m,T]$. Therefore, a dimensional-uniform distribution achieves the capacity of the channel. 
\end{proof}

\begin{proof}[Proof of Lemma~\ref{lem:P2P_MutualInfo_Subspace_Final}]
Assuming an optimal input probability distribution of the form (\ref{eq:optimal_dist_form}), 
the probability of receiving a specific subspace $\Pi_Y=\pi_y$ at the 
receiver can be written as
\begin{align*}
P_{\Pi_Y}(\pi_y) &= \sum_{\pi_x\in\wt{\Xca}} P_{\Pi_Y|\Pi_X}(\pi_y|\pi_x)P_{\Pi_X}(\pi_x) \nonumber\\
&= \sum_{\begin{subarray}{c} \pi_x\in\wt{\Xca},\\ \pi_y\sqsubseteq\pi_x \end{subarray}} \psi(T,n,\pi_y) q^{-nd_x} \frac{\alpha_{d_x}}{\gaussnum{T}{d_x}}.
\end{align*}
Splitting the summation into two, we can write
\begin{align}\label{eq:P_subY_middle1}
P_{\Pi_Y}(\pi_y) &= \psi(T,n,\pi_y) \sum_{d_x=d_y}^{\min[m,T]} \sum_{\begin{subarray}{c} \pi_x\in\wt{\Xca},\\ \dim(\pi_x)=d_x,\\ \pi_y\sqsubseteq\pi_x \end{subarray}} \frac{q^{-nd_x} \alpha_{d_x}}{\gaussnum{T}{d_x}},
\end{align}
where $d_y=\dim(\pi_y)$. Using the following result, 
Lemma~\ref{lem:num_pi_contain_pi0}, we can replace the second 
summation in \eqref{eq:P_subY_middle1}.

\begin{lemma}\label{lem:num_pi_contain_pi0}
Let $\pi_y$ be a fixed subspace of $\mathbb{F}_q^T$ with dimension $d_y$. Then 
the number of different subspaces $\pi_x\in\mathbb{F}_q^T$ with dimension 
$d_x$, $d_y\le d_x\le T$, that contain $\pi_y$ is equal to 
$\gaussnum{T-d_y}{d_x-d_y}$.
\end{lemma}

\begin{proof}
This lemma can be proved by applying \cite[Lemma~2]{GadYan-IT10} with proper choice
of the parameters.
\end{proof}

Using Lemma~\ref{lem:num_pi_contain_pi0} we can rewrite \eqref{eq:P_subY_middle1} as
\begin{align}
P_{\Pi_Y}(\pi_y) &= \psi(T,n,\pi_y) \sum_{d_x=d_y}^{\min[m,T]} \gaussnum{T-d_y}{d_x-d_y} \frac{q^{-nd_x} \alpha_{d_x}}{\gaussnum{T}{d_x}} \nonumber\\
&\stackrel{(a)}{=} \frac{\psi(T,n,\pi_y)}{\gaussnum{T}{d_y}} \sum_{d_x=d_y}^{\min[m,T]} \gaussnum{d_x}{d_y} q^{-nd_x} \alpha_{d_x} \nonumber \\
&= \frac{\psi(n,d_y)}{\gaussnum{T}{d_y}} \sum_{d_x=d_y}^{\min[m,T]} \gaussnum{d_x}{d_y} q^{-nd_x} \alpha_{d_x},
\label{eq:P_subY_simplified}
\end{align}
where $(a)$ follows from the following result, Lemma~\ref{lem:GaussianNum_Relation1}.

\begin{lemma}\label{lem:GaussianNum_Relation1}
The following relation for the Gaussian number holds \cite{Gab-PIT85,And-Book76}
\[
\gaussnum{T-d_y}{d_x-d_y} \gaussnum{T}{d_y} = \gaussnum{T}{d_x} \gaussnum{d_x}{d_y}.
\]
\end{lemma}

Now we can simplify the mutual information $I(\Pi_X;\Pi_Y)$ in 
\eqref{eq:mutual_information_s_1} as follows. Using \eqref{eq:P2P_channel_transfer_prob_2}, 
\eqref{eq:optimal_dist_form}, and \eqref{eq:P_subY_simplified} for $I(\Pi_X;\Pi_Y)$ 
we can write 
\begin{align*}
I(\Pi_X;\Pi_Y) &= \sum_{\pi_x\in\wt{\Xca},\pi_y\in\wt{\Yca}}  P_{\Pi_X}(\pi_x) P_{\Pi_Y|\Pi_X}(\pi_y|\pi_x)\log_2\left(\frac{P_{\Pi_Y|\Pi_X}(\pi_y|\pi_x)}{P_{\Pi_Y}(\pi_y)}\right) \\
&= \sum_{d_x=0}^{\min[m,T]} \sum_{d_y=0}^{\min[n,d_x]} \sum_{\begin{subarray}{c} \pi_x\in\wt{\Xca},\\ \dim(\pi_x)=d_x \end{subarray}} \sum_{\begin{subarray}{c} \pi_y\in\wt{\Yca},\\ \dim(\pi_y)=d_y,\\ \pi_y\sqsubseteq\pi_x \end{subarray}} \frac{\alpha_{d_x} \psi(n,d_y) q^{-nd_x}}{\gaussnum{T}{d_x}} \log_2\left(\frac{q^{-nd_x}}{f(d_y)} \right),
\end{align*}
where 
\begin{equation}
f(d_y)\triangleq\frac{P_{\Pi_Y}(\pi_y)}{\psi(n,d_y)}=\frac{1}{\gaussnum{T}{d_y}} \sum_{d_x=d_y}^{\min[m,T]} \gaussnum{d_x}{d_y} q^{-nd_x} \alpha_{d_x},
\end{equation}
because $P_{\Pi_Y}(\pi_y)$ only depends on $d_y$. Now observe that the two 
inner most summations depend on $\pi_x$ and $\pi_y$ only through their dimensions.
So we can write
\begin{align*}
I(\Pi_X;\Pi_Y) =& \sum_{d_x=0}^{\min[m,T]} \alpha_{d_x}q^{-nd_x} \sum_{d_y=0}^{\min[n,d_x]} \psi(n,d_y)\gaussnum{d_x}{d_y} \log_2\left(\frac{q^{-nd_x}}{f(d_y)} \right) .
\end{align*}
Then using Lemma~\ref{lem:psi_Recursive} in \S\ref{sec:PrelimLemma} we can 
further simplify the mutual information and write
\begin{align}
I(\Pi_X;\Pi_Y) =& -\sum_{d_x=0}^{\min[m,T]} \alpha_{d_x} nd_x  \log_2{q} \nonumber\\
& - \sum_{d_x=0}^{\min[m,T]} \alpha_{d_x} q^{-nd_x} \sum_{d_y=0}^{\min[n,d_x]} \psi(n,d_y)\gaussnum{d_x}{d_y} \log_2(f(d_y)), 
\end{align}
that is the assertion of Lemma~\ref{lem:P2P_MutualInfo_Subspace_Final}.
\end{proof}

\begin{proof}[Proof of Lemma~\ref{lem:I_Derivative_Simplified}]
By taking the partial derivative of the mutual information with 
respect to $\alpha_k$, we have that
\begin{align*}
I'_k \triangleq& \frac{\partial I(\Pi_X;\Pi_Y)}{\partial \alpha_k} \nonumber\\
=& -nk\log_2{q} - \sum_{d_y=0}^{\min[n,k]} \psi(n,d_y) \gaussnum{k}{d_y} q^{-nk} \log_2\left(f(d_y) \right)\nonumber\\
& - \sum_{d_x=0}^{\min[m,T]} \alpha_{d_x} \sum_{d_y=0}^{\min[n,d_x,k]} \psi(n,d_y) \gaussnum{d_x}{d_y} q^{-nd_x} \frac{\gaussnum{k}{d_y} q^{-nk} \log_2e}{\gaussnum{T}{d_y} f(d_y)}. \nonumber\\
\end{align*}
\begin{align*}
I'_k =& -nk\log_2{q} - \sum_{d_y=0}^{\min[n,k]} \psi(n,d_y) \gaussnum{k}{d_y} q^{-nk} \log_2\left(f(d_y) \right)\nonumber\\
& - \sum_{d_y=0}^{\min[n,k]} \frac{\gaussnum{k}{d_y} \psi(n,d_y) q^{-nk}}{f(d_y)} \underbrace{\sum_{d_x=d_y}^{\min[m,T]} \alpha_{d_x} \frac{\gaussnum{d_x}{d_y}}{\gaussnum{T}{d_y}} q^{-nd_x} }_{f(d_y)} \log_2e \nonumber\\
\stackrel{(a)}{=} & -nk\log_2{q} - \sum_{d_y=0}^{\min[n,k]} \psi(n,d_y) \gaussnum{k}{d_y} q^{-nk} \log_2\left(f(d_y) \right) -\log_2{e},
\end{align*}
where to derive $(a)$ we use Lemma~\ref{lem:psi_Recursive} in \S\ref{sec:PrelimLemma}. 
\end{proof}

\begin{proof}[Proof of Lemma~\ref{lem:DominTerm_DerivativeI}]
For convenience we rewrite \eqref{eq:f_Order_Approx} again
\begin{equation}\label{eq:f_Order_Approx_apndx}
\log_2\left(f(d_y)\right) = -d_yT\log_2{q} + O(q^{-1}) + \log_2\left(\sum_{d_x=d_y}^{\min[m,T]} q^{-(n-d_y)d_x} \alpha_{d_x} \right).
\end{equation}
We prove the assertion in two steps for every $k$. First, let us assume that the $\alpha_i$'s are such that we have $\log_2{\left(f(\min[n,k])\right)}=o(q)$. Then using \eqref{eq:f_Order_Approx_apndx} one can conclude that
\[
\sum_{d_x=\min[n,k]}^{\min[m,T]} q^{-(n-d_y)d_x} \alpha_{d_x} = 2^{-o(q)},
\]
so we should have $\alpha_i=2^{-o(q)}$ for $\min[n,k]\le i\le\min[m,T]$. We know that $0\le\alpha_i\le 1$, and $\sum_{i=0}^{\min[m,T]} \alpha_i=1$, so $\exists j:\ \alpha_j=\Omega(1)$. So we can deduce that
\[
\log_2(f(d_y)) = \left\{\begin{array}{ll}
o(q) & j< d_y\le\min[n,k],\\
\Theta(\log{q}) & 0\le d_y\le j,
\end{array} \right.
\]
where $j$, $0\le j\le\min[n,k]$, is the largest index such that $\alpha_j=\Omega(1)$. So in this case the dominating term in the summation of \eqref{eq:I_Derivative_midl1} is the one obtained for $d_y=\min[n,k]$ because the order difference between each term inside the summation of \eqref{eq:I_Derivative_midl1} is at least of order $\Theta(q)$.

Now, for the second case, let us assume that the $\alpha_i$'s are such that we have $\log_2{\left(f(\min[n,k])\right)}=\Omega(q)$. We will show that this assumption leads to a contradiction. Using \eqref{eq:f_Order_Approx_apndx} we can write
\[
\sum_{d_x=\min[n,k]}^{\min[m,T]} q^{-(n-d_y)d_x} \alpha_{d_x} = 2^{-\Omega(q)},
\]
so we should have $\alpha_i=2^{-\Omega(q)}$ for $\min[n,k]\le i\le\min[m,T]$.  As before, we find the asymptotic behavior of $\log_2(f(d_y))$ for different values of $d_y$ but in this case we should make finer regimes for $\log_2(f(d_y))$. The asymptotic behavior of $\alpha_i$, $0\le i\le\min[n,k]$, is either $2^{-\Omega(q)}$ or $2^{-o(q)}$. So we can write
\[
\log_2(f(d_y)) = \left\{\begin{array}{ll}
\Omega(q) & l< d_y \le\min[n,k],\\
o(q) & j< d_y\le l,\\
\Theta(\log{q}) & 0\le d_y\le j,
\end{array} \right.
\]
where $l$, $0\le l\le\min[n,k]$, is the largest index such that $\alpha_i=2^{-o(q)}$ which means that $\alpha_i=2^{-\Omega(q)}$ for $l<i\le\min[m,T]$. As before $j$, $0\le j\le\min[n,k]$, is the largest index such that $\alpha_j=\Omega(1)$.
Now we check the Kuhn-Tucker conditions, \eqref{eq:KuhnTuckerCond_I}, for $I'_k$ and $I'_j$. From the above argument we have that $I'_k=\Omega(q)$ and $I'_j=\Theta(\log{q})$. We know that $\alpha_j=\Omega(1)>0$, so we have
$I'_j=\Theta(\log{q}) = \lambda$.
On the other hand, we have
$I'_k=\Omega(q)\le\lambda$,
which is a contradiction implying the second case cannot occur. This completes the proof.
\end{proof}

\begin{proof}[Proof of Lemma~\ref{lem:N}]
There are $\gaussnum{d_1}{d_{12}}\=q^{d_{12}(d_1-d_{12})}$ different choices for the intersection of $\pi_1$ and $\pi_2$. We have to choose $d_2-d_{12}$ basis vectors for the rest of the subspace. This can be done in 
\begin{align*}
&\frac{\left(q^T-q^{d_1}\right)\left(q^T-q^{d_1+1}\right) \dots \left(q^T-q^{d_1+d_2-d_{12}-1}\right) }
{\left(q^{d_2}-q^{d_{12}}\right) \left(q^{d_2}-q^{d_{12}+1}\right) \dots \left(q^{d_2}-q^{d_2-1}\right)}
\= q^{(d_2-d_{12})(T-d_2)}
\end{align*}
ways. So we have $n(d_1,d_2,d_{12}) \= q^{d_{12}(d_1-d_{12}) +
  (d_2-d_{12})(T-d_2)}$. 
The proof follows from the
results in \cite[Lemma~2]{GadYan-IT10}, by proper choice of parameters.  Independently, an alternate proof of this lemma
  appeared in our paper \cite{MoJaDiFr-ITW09}.
\end{proof}

\begin{proof}[Proof of Lemma~\ref{lem:A}]
Define $\pi=\pi_1 + \pi_2$, where $\dim(\pi)=\dim(\pi_1) +\dim(\pi_2) - \dim(\pi_1\cap \pi_2)=d_1+d_2-d_{12}\triangleq d$. The proof of this lemma is similar to that of Lemma~\ref{lem:N}, unless we can only choose the last $d_2-d_{12}$ basis vectors from $\pi$ instead of $\Fbb_q^T$. Therefore replacing $T$ in Lemma~\ref{lem:N} with $d$, we have
$a(\pi_1,\pi_2)\=q^{d_{12}(d_1-d_{12}) + (d_2-d_{12})(d-d_2)}=q^{d_2(d_1-d_{12})}$.
\end{proof}

\begin{proof}[Proof of Lemma~\ref{lem:corner_col}]
Let $(R_1,R_2)$ be a corner point of the region $\mathcal{R}_{\mathrm{col}}$. Since $\mathcal{R}_{\mathrm{col}}$ is the convex hull of a set of primitive regions, there should exist a primitive region $\mathcal{R}(d_1,d_2)$ which contains $(R_1,R_2)$ as a corner point, \emph{i.e.,} 
\begin{align*}
\exists (d_1,d_2) \in \mathcal{D}_{\mathrm{col}}, \quad  (R_1,R_2)=(R_1(d_1,d_2),R_2(d_1,d_2)).
\end{align*}
We will show that any point $(R_1(d_1,d_2),R_2(d_1,d_2))$ is dominated by the segment connecting $(R_1(d_1+1,d_2),R_2(d_1+1,d_2))$ and $(R_1(d_1,d_2+1),R_2(d_1,d_2+1))$. In order to show that, we have to prove that there exists some $\lambda \in [0,1]$, such that 
\begin{align} \nonumber
R_1(d_1,d_2)&< \lambda R_1(d_1+1,d_2) + (1-\lambda) R_1(d_1,d_2+1),\\
R_2(d_1,d_2)&< \lambda R_2(d_1+1,d_2) + (1-\lambda) R_2(d_1,d_2+1).\label{con2}
\end{align}
After a little simplification, \eqref{con2} can be rewritten as 
\begin{align*}
\lambda [T-d_1-d_2-1] &< d_1,\\
(1-\lambda) [T-d_1-d_2-1] &< d_2,
\end{align*} 
\begin{align*}
\mbox{or } \quad \frac{d_1}{T-1-d_1-d_2}< \lambda < \frac{T-1-d_1-2d_2}{T-1-d_1-d_2}.
\end{align*}
The last two inequalities can be satisfied for some choice of $\lambda$ if and only if $d_1+d_2< (T-1)/2$.  Therefore, if we have $d_1<m_1$, $d_2< m_2$, and $d_1+d_2< (T-1)/2$ for some  $(d_1,d_2)\in\mathcal{D}_{\mathrm{col}}$, then  $(d_1+1,d_2)$ and $(d_1,d_2+1)$ also belong to $\mathcal{D}_{\mathrm{col}}$, and hence, $(R_1(d_1,d_2),R_2(d_1,d_2))$ is an interior point, and cannot be on the boundary of the region. Eliminating such $(d_1,d_2)$ from $\mathcal{D}_{\mathrm{col}}$, we get $\widetilde{\mathcal{D}}$.

It is also easy to show that all of the rate pairs corresponding to $(d_1,d_2)\in\widetilde{\mathcal{D}}$ are on the boundary of $\Rca_{\mathrm{col}}$. This can be done by comparing the slope of the connecting segment for two consecutive points (according to the order they are appeared in $\widetilde{\mathcal{D}}$). The slopes are
\begin{align*}
&\mathcal{S} \{(R_1(t, m_2),R_2(t,m_2)); (R_1(t+1,m_2),R_2(t+1,m_2))\} \\
& \hspace{1cm}= -\frac{m_2}{T-2t-m_2-1}\qquad \textrm{for $0\leq t \leq m_1$}\\
&\mathcal{S} \{(R_1(m_1, t),R_2(m_1,t)); (R_1(m_1,t-1),R_2(m_1,t-1)) \}\\
& \hspace{1cm}= -\frac{T-2t-m_1-1}{m_1} \qquad \textrm{for $1\leq t \leq m_2$}.
\end{align*}
It is easy to check that all the slopes are negative and they are in a decreasing order. Therefore, no point in the set $\widetilde{\mathcal{D}}$ can be an interior point. 
\end{proof}

\begin{proof}[Proof of Lemma~\ref{lem:intersect_col_coop}]
Note that $\Rca_{\mathrm{col}} \nsubseteq \Rca_{\mathrm{coop}}$ implies $m_1+m_2>n$. Since $\Rca_{\mathrm{col}}$ is a convex region, its boundary intersects with the line $R_1+R_2=n(T-n)\log_2{q}$ in exactly two points (it cannot be only one point, otherwise it would be inside of $\Rca_{\mathrm{coop}}$). It is easy to verify that the rate points corresponding to  $(d_1,d_2)=((n-m_2)^+, \min[m_2,n])$ and $(d_1,d_2)=(\min[m_1,n],(n-m_1)^+)$ lie on both the boundary of $\Rca_{\mathrm{col}}$ and the line $R_1+R_2=n(T-n)\log_2{q}$. Therefore this line cannot intersect with the boundary of $\Rca_{\mathrm{col}}$ in any other point.
\end{proof}

\section{Extension to Packet Erasure Networks}\label{sec:apndx2}
Let us write the capacity for the erasure case as follows
\begin{align*}
C_e &= \max_{P_X} I(X;Y,N)\nonumber\\
&= \max_{P_X} \left[ I(X;N) + I(X;Y|N) \right]\nonumber\\
&\stackrel{(a)}{=} \max_{P_X} I(X;Y|N)\nonumber\\
&= \max_{P_X} \Expc{N}{I(X;Y)},
\end{align*}
where (a) follows from the independence of input distribution $P_X$ 
and the distribution of the number of received packets $P_N$. 

\noindent\textbf{The Upper Bound:}\\
We can write an upper bound for $C_e$ as follows
\begin{align*}
C_e &= \max_{P_X} \Expc{N}{[I(X;Y)]} \nonumber\\
&\le \Expc{N}{\max_{P_X} I(X;Y)} \nonumber\\
&= \Expc{N}{i^*(T-i^*)\log_2{q}},
\end{align*}
where $i^*=\min[m,N,\lfloor T/2 \rfloor]$. From here on let us assume that 
$m\le\lfloor T/2 \rfloor$. We thus have that $i^*=N$ and we can write
\begin{align*}
C_e &\le \Expc{N}{ N(T-N)\log_2{q}}.
\end{align*}
Let us define $\mu_1\triangleq\Expc{N}{N}$ and $\mu_2\triangleq\Expc{N}{N^2}$ so we can write
\begin{equation*}
C_e \le \left(\mu_1 T-\mu_2\right)\log_2{q}.
\end{equation*}

\noindent\textbf{The Lower Bound:}\\
For the lower bound we can write
\begin{align*}
C_e &= \max_{P_X} \Expc{N}{[I(X;Y)]} \nonumber\\
&\ge  \Expc{N}{I(X;Y)}_{\text{for some $P_X$}} \nonumber\\
&=  \Expc{N}{I(\Pi_X;\Pi_Y)}_{\text{for some $P_{\Pi_X}$}}.
\end{align*}

From \eqref{eq:P2P_MutualInfo_Subspace_Final} we know that we can write
\begin{align*}
I(\Pi_X;\Pi_Y) =& -\sum_{d_x=0}^{\min[m,T]} \alpha_{d_x} Nd_x  \log_2{q} \nonumber\\
& - \sum_{d_x=0}^{\min[m,T]} \alpha_{d_x} q^{-Nd_x} \sum_{d_y=0}^{\min[N,d_x]} \psi(N,d_y)\gaussnum{d_x}{d_y} \log_2(f(d_y)),
\end{align*}
where 
\begin{equation*}
f(d_y)\triangleq \frac{1}{\gaussnum{T}{d_y}} \sum_{d_x=d_y}^{\min[m,T]} \gaussnum{d_x}{d_y} q^{-Nd_x} \alpha_{d_x}.
\end{equation*}

Now assume that $m\le \lfloor T/2\rfloor$ and choose the input 
distribution to be $\alpha_k=1$ for some $0\le k\le m$ and 
$\alpha_i=0$ for all $i\neq k$. Then for this input distribution 
we have
\begin{align*}
I(\Pi_X;\Pi_Y) =& -kN\log_2{q} - q^{-kN} \sum_{d_y=0}^{\min[N,k]} \psi(N,d_y)\gaussnum{k}{d_y} \log_2(f(d_y))\nonumber\\
=& -kN\log_2{q} - q^{-kN} \sum_{d_y=0}^{\min[N,k]} \psi(N,d_y)\gaussnum{k}{d_y} \log_2(f(d_y)).
\end{align*}
Then assuming $q$ is large we may approximate the above mutual 
information as follows
\begin{align*}
I(\Pi_X;\Pi_Y) &\approx -kN\log_2{q} - \sum_{d_y=0}^{\min[N,k]} q^{-(N-d_y)(k-d_y)} \log_2(f(d_y)).
\end{align*}
The term $(N-d_y)(k-d_y)$ in the summation is maximized for $d_y=\min[N,k]$ 
and because we had shown before in Lemma~\ref{lem:DominTerm_DerivativeI} 
that $\log_2(f(d_y))=\Theta(\log{q})$, we can write
\begin{align*}
I(\Pi_X;\Pi_Y) &\approx -kN\log_2{q} - \log_2(f(\min[N,k])) \nonumber\\
&\approx -kN\log_2{q} - \log_2\left( q^{\min[N,k](k-T)-Nk}  \right)\nonumber\\
&= \min[N,k](T-k)\log_2{q}.
\end{align*}

So by choosing $k=m$ we can write the lower bound for $C_e$ as follows
\begin{align*}
C_e &\ge \Expc{N}{I(\Pi_X;\Pi_Y)}_{\text{for some $P_{\Pi_X}$}}\nonumber\\
&\approx  \Expc{N}{N(T-m)\log_2{q}}\nonumber\\
&= \mu_1\left(T-m\right)\log_2{q}.
\end{align*}
%
%


\begin{biographynophoto}{Mahdi Jafari Siavoshani} received the Bachelor degree
in Communication Systems with a minor in Applied Physics at Sharif
University of Technology, Tehran, Iran, in 2005.  He was awarded an
Excellency scholarship from EPFL, Switzerland, to study a master
degree in Communication System finished in 2007.  He is currently a
PhD student at the same university.  His research interests include
network coding, coding and information theory, wireless
communications, and signal processing.
\end{biographynophoto}

\begin{biographynophoto}{Soheil Mohajer} received the B.S. degree in electrical engineering 
from the Sharif University of Technology, Tehran, Iran, in 2004, and
the M.S. degrees in communication systems from Ecole Polytechnique
Fédérale de Lausanne (EPFL), Lausanne, Switzerland, in 2005.  He
completed his Ph.D. at EPFL in September 2010, and since October 2010
is a post-doctoral researcher at Princeton University. His fields of
interests are multiuser information theory, network coding theory, and
wireless communication.
\end{biographynophoto}

\begin{biographynophoto}{Christina Fragouli} is a
tenure-track Assistant Professor in the School of Computer and
Communication Sciences, EPFL, Switzerland. She received the
B.S. degree in Electrical Engineering from the National Technical
University of Athens, Athens, Greece, in 1996, and the M.Sc. and
Ph.D. degrees in electrical engineering from the University of
California, Los Angeles, in 1998 and 2000, respectively. She has
worked at the Information Sciences Center, AT\&T Shannon Labs, Florham Park New
Jersey, and the National University of Athens.  She also visited Bell
Laboratories, Murray Hill, NJ, and DIMACS, Rutgers University. From
2006 to 2007, she was an FNS Assistant Professor in the School of
Computer and Communication Sciences, EPFL, Switzerland. She served as
an editor for IEEE Communications Letters. She is currently serving as
an editor for IEEE Transactions on Information Theory, IEEE
Transactions on Communications, Elsevier Computer Communications and
IEEE Transactions on Mobile Computing.  She was the technical co-chair
for the 2009 Network coding symposium in Lausanne and has served on
program commmittees of several conferences.  She received the
Fulbright Fellowship for her graduate studies, the Outstanding
Ph.D. Student Award 2000-2001, UCLA, Electrical Engineering
Department, the Zonta award 2008 in Switzerland, and the Young
Investigator ERC starting grant in 2009. Her research interests are in
network information flow theory and algorithms, network coding, and
connections between communications and computer science.
\end{biographynophoto}

\begin{biographynophoto}{Suhas N. Diggavi} received the B. Tech. degree in electrical
engineering from the Indian Institute of Technology, Delhi, India, and
the Ph.D. degree in electrical engineering from Stanford University,
Stanford, CA, in 1998.

After completing his Ph.D., he was a Principal Member Technical Staff
in the Information Sciences Center, AT\&T Shannon Laboratories, Florham
Park, NJ. After that he was on the faculty at the School of Computer
and Communication Sciences, EPFL, where he directed the Laboratory for
Information and Communication Systems (LICOS).  He is currently a
Professor, in the Department of Electrical Engineering, at the
University of California, Los Angeles.  His research interests include
wireless communications networks, information theory, network data
compression and network algorithms.

He is a recipient of the 2006 IEEE Donald Fink prize paper award, 2005
IEEE Vehicular Technology Conference best paper award and the Okawa
foundation research award.  He is currently an editor for ACM/IEEE
Transactions on Networking and IEEE Transactions on Information
Theory. He has 8 issued patents.
\end{biographynophoto}

\end{document}